%% file: main.tex
\documentclass[preprint,12pt]{elsarticle}

\usepackage{amsmath}
\usepackage{amssymb}
\usepackage{amsthm}
\usepackage{mathrsfs}
\usepackage{graphicx}
\usepackage{hyperref}
\usepackage{lineno}
\usepackage{caption}
\usepackage{subcaption}
\usepackage{booktabs}
\usepackage{multirow}
\usepackage[inline]{enumitem}

\newtheorem{lemma}{Lemma}

\newtheorem{remark}{Remark}

\input{preamble.tex}

\makeatletter
\def\ps@pprintTitle{%
  \let\@oddhead\@empty
  \let\@evenhead\@empty
  \def\@oddfoot{\reset@font\hfil\thepage\hfil}
  \let\@evenfoot\@oddfoot
}
\makeatother

\begin{document}
\begin{frontmatter}

\title{Diagonalization of the Landau Hamiltonian with a Periodic Potential via a Galerkin Projection Method and Applications to Topological Band Properties}
\author[1]{Rafael Antonio Lainez Reyes\corref{cor1}}
\ead{rafael-antonio.lainez-reyes@mathematik.uni-stuttgart.de}
\cortext[cor1]{Corresponding author}

\author[1]{Benjamin Stamm}
\ead{benjamin.stamm@mathematik.uni-stuttgart.de}

\affiliation[1]{organization={Institute of Applied Analysis and Numerical Simulation, University of Stuttgart},
            addressline={Pfaffenwaldring 57},
            city={Stuttgart},
            postcode={70569},
            state={Baden-Württemberg},
            country={Germany}}

\author[2]{Hans Peter Büchler}
\ead{buechler@theo3.physik.uni-stuttgart.de}
\affiliation[2]{organization={Institute for Theoretical Physics III, University of Stuttgart},
            addressline={Pfaffenwaldring 57},
            city={Stuttgart},
            postcode={70569},
            state={Baden-Württemberg},
            country={Germany}}

\begin{abstract}
In this paper, we study the numerical computation of eigenfunctions of the Landau Hamiltonian with a periodic potential. We propose a spectral method using the eigenfunctions of the Landau operator that allows us to numerically compute the band structure for potentials of different strengths, including the well-studied weak and strong potentials, and additionally systems in the intermediate regime. We then apply the method to simulations across the full range of coupling constants from weak to strong coupling and present an analysis of the band structure as a function of the potential strength. 
\end{abstract}

\begin{highlights}
\item Spectral method for computing eigenfunctions of the Landau Hamiltonian with a periodic potential.
\item Numerical computations of the band structure for weak, strong, and intermediate potential strengths. 
\item Analysis of the band structure as a function of the potential strength.
\end{highlights}


\end{frontmatter}

\input{Introduction.tex}
\input{Theory.tex}
\input{Numerics.tex}
\input{Results.tex}
\input{Conclusions.tex}
\FloatBarrier
\input{Data_availability_acknowledgements.tex}

\bibliographystyle{elsarticle-num} 
\bibliography{references}

\end{document}

%% file: preamble.tex
\usepackage{bm}
\usepackage{comment}
\usepackage{adjustbox}
\usepackage{xspace}
\usepackage[colorinlistoftodos]{todonotes}
\usepackage{pdfpages}
\usepackage{xparse}
\usepackage{placeins}
\usepackage{ifthen} 
\usepackage{dsfont} 
\usepackage{xcolor} 

\newcommand{\bbC}{\mathbb{C}}
\newcommand{\bbR}{\mathbb{R}}
\newcommand{\bbN}{\mathbb{N}}

\newcommand{\bbZ}{\mathbb{Z}}

\newcommand{\mcC}{\mathcal{C}}

\newcommand{\mcL}{\mathcal{L}}

\newcommand{\mcV}{\mathcal{V}}

\newcommand{\mcX}{\mathcal{X}}

\newcommand{\bfa}{\mathbf{a}}

\newcommand{\bfc}{\mathbf{c}}
\newcommand{\bfe}{\mathbf{e}}
\newcommand{\bfk}{\mathbf{k}}
\newcommand{\bfl}{{\bm{\ell}}}

\newcommand{\bfx}{\mathbf{x}}

\newcommand{\bfxi}{{\bm{\xi}}}

\newcommand{\bfA}{\mathbf{A}}

\newcommand{\bfH}{\mathbf{H}}

\newcommand{\bfS}{\mathbf{S}}

\newcommand{\iu}{\mathrm{i}}

\newcommand{\inpro}[2]{\langle#1,#2\rangle}

\newcommand{\norm}[1]{\lVert#1\rVert}
\newcommand{\set}[1]{\left\lbrace#1\right\rbrace}

\newcommand{\conjz}{\overline{z}}

\newcommand{\Cnumber}{\mathcal{C}}

\newcommand{\smoothFunctions}[3][]{\ifthenelse{\equal{#1}{}}{\mathcal{C}^{#2}}{\mathcal{C}_{#1}^{#2}}(#3)}

\newcommand{\dist}[1][]{\ifthenelse{\equal{#1}{}}{\mathbb{D}}{#1_{\mathbb{D}}}}



%% file: Introduction.tex
\section{Introduction\label{sec:Introduction}}

Electronic systems in the presence of magnetic fields exhibit a wide range of interesting phenomena. A system subject to a uniform perpendicular magnetic field admits a discrete set of energy levels known as Landau levels~\cite{landauDiamagnetismusMetalle1930}, which, in the absence of an external potential, are infinitely degenerate. We consider such a system, subject to a periodic potential $V(\bfx)$. While this problem is of fundamental interest as a paradigmatic quantum mechanical problem, it has recently attracted increased interest due to experiments with cold atomic gases in rotating traps, which give rise to an effective magnetic field \cite{Zwierlein2021,Zwierlein2024}, as well as twisted bilayer graphene \cite{Bernevig2022}. For simplicity, we focus on a square lattice with period $\bfa=(a,a)\in \bbR^2$, inducing a Bravais lattice $\mcL=a\bbZ\times a\bbZ$ with a corresponding unit cell $\Omega$. The generalization to arbitrary lattices is straightforward. Two limiting cases of this system are well-studied: the tight-binding limit, corresponding to a strong external potential relative to the magnetic field, which yields the famous fractal spectrum known as Hofstadter's butterfly~\cite{hofstadterEnergyLevelsWave1976} (see also~\cite{fuchsLandauLevelsQuasicrystals2018}); and the weak-potential limit, corresponding to a strong magnetic field relative to the external potential, often explored via perturbative approaches~\cite{pfannkucheTheoryMagnetotransportTwodimensional1992, fuchsLandauLevelBroadening2019}. Already in these limiting cases, the system displays a rich and complex band structure. 

The system introduced above is governed by the Hamiltonian
\begin{equation}
    H = L(\bfA) + V(\bfx) = \frac{1}{2}(-\iu\nabla -\bfA(\bfx))^2 + V(\bfx), \quad \bfx \in \bbR^2\label{eq: Hamiltonian},
\end{equation}
where $\bfA(\bfx) = \frac{B}{2}[-x_2, x_1]$ is the magnetic vector potential in the symmetric gauge, inducing a constant magnetic field $B\bfe_3$ perpendicular to the $xy$-plane. Here, we express energies by the characteristic kinetic energy within the lattice $E_r = \hbar^2/m a^2$, lengths in terms of the lattice period $a$, and introduce an effective magnetic field $B$. For particles with charge $q$, it is related to the microscopic magnetic field $\mathcal{B}$ via $B= q \mathcal{B} a^2/\hbar $, while for particles in a rotating trap with rotation frequency $\omega$, it reduces to $B= 2 m \omega a^2/\hbar$. From now on, we will work using a unit lattice period, that is, $a=1$. The generalization to arbitrary lattice periods is straightforward.

It is known that $H$ does not commute with the standard translation operators $\set{T_\bfl}_{\bfl\in\mcL}$ associated with the lattice $\mcL$. This precludes the direct application of standard Bloch theory~\cite{ashcroftSolidStatePhysics1976, reedmichaelMethodsModernMathematical1972}. However, under the \textit{integer flux condition}, where the magnetic flux over a unit cell $\Omega$ satisfies $\int_{\Omega}\nabla\times \bfA(\bfx)\,\mathrm{d}\bfx = 2\pi \beta$ for $\beta \in \bbZ$, an extension of Bloch theory pioneered by Zak~\cite{zakMagneticTranslationGroup1964} can be applied. In cases of rational flux $\beta = p/q$, the theory remains applicable through a supercell formulation on a modified lattice $\mathcal{L}'$. A rigorous mathematical description of magnetic Bloch theory can be found in the works of Iwatsuka~\cite{iwatsukaSchrodingerOperatorsMagnetic1990} and Yoshitomi~\cite{yoshitomiSchrodingerOperatorsPeriodic1997}.

This magnetic Bloch theory allows the characterization of $H$, which is defined over $\bbR^2$, by studying a family of Hamiltonians parametrized over the Brillouin zone $\Omega^*$, namely $\{H_\bfk\}_{\bfk\in\Omega^*}$, where each $H_\bfk$ is defined over the bounded unit cell $\Omega$. Determining the resulting band structure is thus equivalent to solving a parametrized eigenvalue problem for $H_\bfk$ across the Brillouin zone. 

While the previously introduced limiting regimes, which are in a sense equivalent~\cite{langbeinTightBindingNearlyFreeElectronApproach1969}, are well-studied in the literature, the intermediate regime, where neither the magnetic field nor the external potential dominates, remains comparatively unexplored. The main complication in the intermediate regime is that the coupling between different Landau levels is no longer negligible, preventing the application of methods based on perturbation theory. In the current work, we propose a numerical method inspired by the plane-wave methods used in Kohn-Sham systems~\cite{cancesNumericalAnalysisPlanewave2012}. Our approach employs a Galerkin framework to approximate the solutions of the parametrized eigenvalue problem. No assumptions on the strength of the potential or the magnetic field are necessary, which allows us to treat the full range of parameters, encompassing the weak-potential and tight-binding limits as well as the intermediate regime. 

The organization of the paper is as follows: in Section~\ref{sec:Theory}, we introduce the previously mentioned parametrized eigenvalue problem. Section~\ref{sec:Numerics} focuses on the Galerkin strategy for the efficient computation of matrix elements, and in Section~\ref{sec: Results}, we explore the resulting band structures.

%% file: Theory.tex
\section{Background\label{sec:Theory}}
    In this section, we briefly review Bloch theory applied to the Landau Hamiltonian. Since these results are well known in the literature, we limit the exposition and refer the reader to~\cite{kohmotoTopologicalInvariantQuantization} and references therein. The building blocks of magnetic Bloch theory are the magnetic translation operators, which we now introduce. Given a lattice vector $\bfl=(\ell_1,\ell_2)\in\mcL$, let $T_\bfl$ be a standard translation operator, that is
    \begin{equation}
        T_\bfl\psi(\bfx)=\psi(\bfx-\bfl).
    \end{equation}
    The magnetic translation operator $S_\bfl$ is defined in terms of the standard translation operator $T_\bfl$ as
    \begin{equation}
        S_\bfl\psi(\bfx)=\exp\left(-\iu B\frac{-\ell_1x_2+x_1\ell_2}{2}\right)T_\bfl\psi(\bfx)\label{eq: Magnetic Translation Operator}.
    \end{equation}
    It is easy to verify that $H(\bfA)S_\bfl=S_\bfl H(\bfA)$ and $L(\bfA)S_\bfl=S_\bfl L(\bfA)$ as operators; however, the only case where $S_{\bfl_1}S_{\bfl_2}=S_{\bfl_2}S_{\bfl_1}$ for any $\bfl_1,\bfl_2\in\mcL$ is when the integer (or rational by extension) flux condition is satisfied.  
    Therefore, from now on, we will assume the integer flux condition is satisfied. The extension to the rational flux case can be carried out without major modifications. Under this assumption, the study of the eigenfunctions of the Landau Hamiltonian reduces to an eigenvalue problem over a unit cell, parametrized over reciprocal unit cells $\Omega$ and $\Omega^*$. For convenience, we choose
    \begin{align}
        \Omega&=[-1,0]\times[-1,0],\\
        \Omega^*&=\left[-\pi, \pi\right]\times\left[-\pi, \pi\right].
    \end{align}
    The parametrized eigenvalue problem in question is
    \begin{align}
        \label{eq: Eigenvalue Problem at k}
        (L(\bfA)+V(\bfx)) \psi_{m,\bfk}(\bfx) & = \lambda_{m,\bfk}\psi_{m,\bfk}(\bfx),\\
        \label{eq:Magnetic boundary condition}
        S_{\bfl}\psi_{m,\bfk}(\bfx) & = \exp(-\iu \bfl\cdot \bfk)\psi_{m,\bfk}(\bfx),
    \end{align}
    where $\bfx\in\Omega,\; \bfk\in\Omega^*,\; \bfl\in\mcL$, and $m\in \bbN_0$ corresponds to the energy level. In order to solve this problem numerically, we need to introduce a basis which satisfies the boundary condition in Eq.~\eqref{eq:Magnetic boundary condition} induced by the magnetic translations, which we will refer to as the magnetic boundary condition. Although it is not strictly necessary, for reasons explained below, an orthonormal basis is desirable.

    \subsection{Eigenfunctions of the Magnetic Translations\label{subsec: Eigenfunctions magnetic Translations}}
        In this section, we recall the construction of the magnetic translation eigenfunctions as presented in \cite{hoW_inftySl_q2Algebras1995} and \cite{dereliBlochWavesNoncommutative2021}. These eigenfunctions will be the building blocks for the numerical method we present in the current work. We begin by obtaining a set of eigenfunctions for the Landau operator on the plane. In order to do so, it is convenient to introduce complex variables. Therefore, for a given $\bfx=(x_1,x_2)\in\bbR^2$, we define
        \begin{equation}
            z = x_1+\iu x_2,\quad \conjz=x_1-\iu x_2.
        \end{equation}
        Associated with $z$ and $\conjz$, we have the Wirtinger derivative operators:
        \begin{align}
            \partial_z &=\frac{1}{2}(\partial_{x_1}-\iu\partial_{x_2})\\
            \partial_{\conjz} &=\frac{1}{2}(\partial_{x_1}+\iu\partial_{x_2}).
        \end{align}
        Using the Wirtinger operators, we can now define the raising and lowering operators
        \begin{align}
            L^+&=-\sqrt{\frac{2}{B}}\partial_z+\frac{1}{2}\sqrt{\frac{B}{2}}\conjz\label{op: creation operator in complex}\\
            L^-&=\sqrt{\frac{2}{B}}\partial_{\conjz}+\frac{1}{2}\sqrt{\frac{B}{2}}z.\label{op: annihilation operator in complex}
        \end{align}
        As expected, the raising and lowering operators allow us to express the Landau operator in a rather convenient way; it is straightforward to verify that
        \begin{equation}
            L(\bfA)=B\left(L^+L^-+\frac{1}{2}\right).
        \end{equation}
        This is a harmonic oscillator, and by standard analysis, we obtain the family of normalized ground-state eigenfunctions,
        \begin{equation}
            \phi_0(z,\conjz)=C F(z)\exp\left(-\frac{B}{4}z\conjz\right)\label{eq: General ground state},
        \end{equation}
        where $F$ is any analytic function and $C$ a normalization constant to be determined, and the corresponding infinitely degenerate eigenvalue is
        \begin{equation}
            E_0 = \frac{B}{2}.
        \end{equation}
        Once the ground states are known, the normalized excited states can be obtained by repeatedly applying the raising operator:
        \begin{equation}\label{eq: Landau Operator in Plane Eigenfunction}
            \phi_m(z,\conjz)= (L^+)^m\phi_0(z,\conjz),
        \end{equation}
        with corresponding energy 
        \begin{equation}
            E_m = B\left(m+\frac{1}{2}\right),\quad m\in\bbN_0.
        \end{equation}
        Since $(L^+)^\dagger=L^-$ and $L^-\phi_{m+1}=(m+1)\phi_m$, it is easy to show that the degeneracy from the ground state propagates to the excited states. Let us now impose the magnetic boundary conditions from Eq.~\eqref{eq:Magnetic boundary condition} over the eigenfunctions of the Landau operator on the plane.
        
        The magnetic translations commute with the Landau operator $L(\bfA)$, therefore they must share a common basis of eigenfunctions. All the eigenfunctions of the Landau operator $L(\bfA)$ are of the form of Eq.~\eqref{eq: Landau Operator in Plane Eigenfunction}, so we are left to choose $F$ in Eq.~\eqref{eq: General ground state} so that $\phi_0$ is an eigenfunction of the magnetic translation operators. The resulting normalized ground-state eigenfunctions are
        \begin{align}
            \phi_{0,\alpha,\beta, \bfk}(z,\conjz)&=N_{k_1}\exp\left(\frac{B}{4}z^2+\iu k_1z-\frac{B}{4}z\conjz\right)\hat\Theta_{\alpha,\beta}(z-\delta(\bfk)),\label{eq: Landau Eigenfunctions at Ground State}\\
            N_{k_1}&=\left(\frac{B}{\pi}\right)^{1/4}\exp\left(-\frac{k_1^2}{2B}\right),\\
            \alpha&\in \set{0,1,\ldots,\beta-1},\\
            \beta&=\frac{B}{2\pi},\\
            \delta(\bfk)&=\frac{k_2-\iu k_1}{B}, \quad\bfk = (k_1,k_2)\in\Omega^*,
        \end{align}
        where $N_{k_1}$ is the normalization constant derived in Lemma~\ref{lem: Normalization Constant} below. The function $\hat\Theta_{\alpha, \beta}$ is a Theta function~\cite{mumfordTataLecturesTheta2007}, given by the series expansion:
        \begin{equation}
            \hat\Theta_{\alpha, \beta}(z)=\sum_{n\in\bbZ}\exp(-\pi \beta(n+\alpha/\beta)^2)\exp(2\pi \iu \beta(n+\alpha/\beta)z).
        \end{equation} 
        The constant $\beta$ is the magnetic flux per unit cell and, by assumption, is an integer. Regarding $\alpha$, for the proposed function to be an eigenfunction,
        \begin{equation}
            \beta\left(n+\frac{\alpha}{\beta}\right)\in\bbZ, \quad \forall n\in\bbZ,
        \end{equation}
        is a necessary condition. This implies $\alpha+n\beta\in \bbZ$ for all $n\in\bbZ$. In other words, $\alpha$ is an integer modulo $\beta$. For simplicity, we choose $\alpha\in\set{0,1,\ldots, \beta-1}$. Translating $\alpha$ by an integer multiple of $\beta$ does not result in a new eigenfunction. On the other hand, for given $\alpha_1,\alpha_2\in \bbZ$ with $0\leq \alpha_1,\alpha_2\leq \beta-1$ and $\alpha_1\neq \alpha_2$, the resulting ground-state functions $\phi_{0,\alpha_1,\beta, \bfk}$ and $\phi_{0,\alpha_2,\beta, \bfk}$ are linearly independent. In contrast to the Landau operator on the plane, the ground state of the Landau operator satisfying the magnetic boundary conditions has a finite degeneracy of degree $\beta$. More can be said about the set $\set{\phi_{0,\alpha,\beta, \bfk}:\; \alpha\in\set{0,1,\ldots,\beta-1}}$.
        \begin{lemma}\label{lem: Normalization Constant}
            The set of functions $\set{\phi_{m,\alpha,\beta,\bfk}:\; m\in\bbN_0,\; \alpha\in\set{0,1,\ldots,\beta-1}}$ defined in Eq.~\eqref{eq: Landau Eigenfunctions at Ground State} is such that each $|\phi_{0,\alpha,\beta, \bfk}(\bfx)|^2$ is periodic with respect to the lattice $\mcL$. As a consequence, it allows us to sensibly define the $L^2$-norm over the unit cell $\Omega$. Moreover, the following properties hold:
            \begin{equation}
                \norm{\phi_{0,\alpha,\beta, \bfk}}_{L^2(\Omega)}^2 = \int_\Omega |\phi_{0,\alpha,\beta, \bfk}(\bfx)|^2 \,\mathrm{d}\bfx =  N^2_{k_1} \sqrt{\frac{\pi}{B}}\exp\left(\frac{k_1^2}{B}\right) =1, \label{eq: normalization constant}
            \end{equation}
            and for $\alpha'\neq \alpha$, 
            \begin{equation}
                \int_\Omega \overline{\phi}_{0,\alpha',\beta, \bfk}(\bfx) \phi_{0,\alpha,\beta, \bfk}(\bfx) \,\mathrm{d}\bfx = 0, \label{eq: orthogonality}
            \end{equation}
            that is, the set is orthogonal in $L^2(\Omega)$.
        \end{lemma}

        \begin{proof}
            Let $\gamma_n(\alpha)=n+\alpha/\beta$, and analogously define $\gamma_{n'}(\alpha')$. Straightforward computations show
            \begin{align}
                &\overline{\phi}_{0,\alpha',\beta, \bfk}(\bfx) \phi_{0,\alpha,\beta, \bfk}(\bfx)/N_{k_1}^2\\= &  \exp(-Bx_2^2-2k_1 x_2) \sum_{n,n'\in \bbZ} C_{n,n'} F_{n,n'}(x_1) G_{n,n'}(x_2) H_{n,n'}(\bfk) 
            \end{align}
            with
            \begin{align}
                C_{n,n'} &= \exp\left(-\pi\beta(\gamma_n(\alpha')^2+\gamma_{n'}(\alpha)^2)\right)\label{eq: phi phi conjugate for orthogonality} \\
                F_{n,n'}(x_1) &= \exp\left(-2\pi \beta\iu (\gamma_n(\alpha')-\gamma_{n'}(\alpha)) x_1\right) \\
                G_{n,n'}(x_2) &= \exp\left(-2\pi \beta (\gamma_n(\alpha')+\gamma_{n'}(\alpha)) x_2\right) \\
                H_{n,n'}(\bfk) &= \exp\left(-(\gamma_n(\alpha')+\gamma_{n'}(\alpha))k_1\right) \exp\left(\iu  (\gamma_n(\alpha')-\gamma_{n'}(\alpha))k_2\right).
            \end{align} 
            Now we integrate over $\Omega$. Since the sums are uniformly convergent, we can exchange them with the integral. Ultimately, the integrand becomes separable. In particular, the integration over $x_1$ yields:
            \begin{equation}
                \int_{-1}^0 F_{n,n'}(x_1) \,\mathrm{d}x_1 = 
                \begin{cases} 
                    1 & \text{if } n=n' \text{ and } \alpha=\alpha', \\ 
                    0 & \text{otherwise.}
                \end{cases}
            \end{equation}
            This implies the integral is zero whenever $\alpha \neq \alpha'$, giving the desired orthogonality property.
            
            Let us now compute the norm; therefore, $\alpha = \alpha'$ and $n = n'$. For simplicity, $\gamma_n(\alpha) = \gamma_{n'}(\alpha') \equiv \gamma_n$. The terms simplify to:
            \begin{align}
                C_{n,n} &= \exp\left( -2\pi\beta \gamma_n^2 \right), \\
                F_{n,n}(x_1) &= 1, \\
                G_{n,n}(x_2) &= \exp\left( -4\pi\beta \gamma_n x_2 \right), \\
                H_{n,n}(\bfk) &= \exp\left( -2 \gamma_n k_1 \right).
            \end{align}
            Hence, after substituting into Eq.~\eqref{eq: phi phi conjugate for orthogonality}, we have 
            \begin{align}
                |\phi_{0,\alpha,\beta, \bfk}&(\bfx)|^2/  N^2_{k_1}
                \\  = & \sum_{n \in \bbZ} \exp\left( -2\pi\beta \gamma_n^2 - 2  \gamma_n k_1 \right)\exp\left( -B x_2^2 - 2\left(k_1  + 2\pi\beta \gamma_n\right) x_2 \right)\\
                = &\exp(k_1^2/B) \sum_{n \in \bbZ} \exp\left( -B \left( x_2 + \gamma_n + \frac{k_1}{B} \right)^2 \right).
            \end{align}
            Finally, we integrate over $\Omega$. Using the change of variable $u = x_2 + n$, the sum of integrals over $[-1, 0]$ becomes a single integral over $\bbR$:
            \begin{align}
                &\norm{\phi_{0,\alpha,\beta, \bfk}(\bfx)}_{L^2(\Omega)}^2 /  N^2_{k_1}\\ = & \exp(k_1^2/B)\sum_{n\in\bbZ}\int_{-1}^0 \exp\left( -B \left( x_2 + \gamma_n + \frac{k_1}{B} \right)^2\right)\,\mathrm{d}x_2\\
                = &\exp(k_1^2/B)\int_\bbR\exp\left(-B\left(u + \frac{\alpha}{\beta}+\frac{k_1}{B}\right)^2\right)\,\mathrm{d}u\\
                = & \sqrt{\frac{\pi}{B}}\exp\left(\frac{k_1^2}{B}\right),
            \end{align}
            as claimed.
        \end{proof}

        The normalized excited states corresponding to energy $E_m=B(m+1/2)$ belong to the eigenspace spanned by the set
        \begin{equation}
            \set{\frac{1}{\sqrt{m!}}(L^+)^m \phi_{0,\alpha,\beta, \bfk}(z,\conjz): \; \alpha\in\set{0,1,\ldots,\beta-1}}.
        \end{equation}
        \begin{remark}\label{remark: normalized raising and lowering operators}
            After normalization, the usual properties of the raising and lowering operators acting over an eigenfunction must be modified. More precisely, for any $m\in\bbN_0$, $\alpha\in\set{0,1,\ldots,\beta-1}$, and $\bfk\in\Omega^*$, we have
            \begin{align}
                \phi_{m+1,,\alpha,\beta, \bfk}&=\frac{1}{\sqrt{m+1}}L^+\phi_{m,\alpha,\beta, \bfk}, \\
                L^-\phi_{m+1,\alpha,\beta, \bfk}&=\sqrt{m+1}\phi_{m,\alpha,\beta, \bfk}.
            \end{align}

        \end{remark}
        In the next section, we will determine explicit expressions for the excited states, which are suitable for numerical simulations. We close this section by remarking that the set $\set{\phi_{m,\alpha,\beta,\bfk}:\; m\in\bbN_0,\; \alpha\in\set{0,1,\ldots,\beta-1}}$ is a basis of $\set{\psi_\bfk\in L^2(\Omega):  S_{\bfl}\psi_{\bfk}(\bfx) = \exp(-\iu \bfl\cdot \bfk)\psi_{\bfk}(\bfx)}$ for each fixed $\bfk\in\Omega^*$.
    
    \subsection{Hall conductivity and Chern number\label{sec:hall_conductivity_and_Chern_number}}
        Let $\set{(\psi_{n,\bfk}, \lambda_{n,\bfk}):\,n\in\bbN_0,\,\bfk\in\Omega^*}$ be the eigenpairs obtained when solving Eq.~\eqref{eq: Eigenvalue Problem at k}, 
        which in the absence of an external potential correspond to the Landau levels. In general, the set of values $\set{\lambda_{n,\bfk}: \bfk\in\Omega^*}$ for each 
        $n\in\bbN_0$ defines the band structure of the system. This band structure can give rise to a non-trivial Chern number $\Cnumber_n$ and, correspondingly, to a 
        quantized Hall conductance \cite{thoulessQuantizedHallConductance1982}. To determine the Chern number of the $n$-th band, we introduce 
        $u_{n,\bfk}(\bfx)=\exp(-\iu\bfx\cdot\bfk)\psi_{n,\bfk}(\bfx)$ and define the Berry connection
        \begin{equation}
        \bfa_n(\bfk) = - \iu \langle u_{n,\bfk} | \nabla_\bfk u_{n,\bfk} \rangle.  
        \end{equation}
        Notice that the Berry connection is a vector (i.e., $\bfa_n(\bfk)=(a_{n,1}(\bfk), a_{n,2}(\bfk))$). Furthermore, this quantity is independent of the choice of the unit cell as
        $\overline{u}_{n,\bfk}(\bfx)\nabla_\bfk u_{n,\bfk}(\bfx)$ is periodic, which follows immediately from the definition of the magnetic translation operator. 
        Here, we perform the analysis for non-degenerate bands. Later, we will determine the Chern number numerically using the Fukui-Hatsugai-Suzuki method~\cite{fukuiChernNumbersDiscretized2005},
        which can then be straightforwardly generalized to degenerate bands. The Chern number $\Cnumber_n$ is defined in terms of the Berry connection and is given by 
        \begin{equation}
                \Cnumber_n = \frac{1}{2 \pi } \int_{\Omega^*} \left[ \partial_{k_1} a_{n,2}(\bfk) -\partial_{k_2} a_{n,1}(\bfk)\right] \,\mathrm{d}\bfk.
        \end{equation}
        For the system with no external potential, the eigenstates of the band structure are given by the basis states $\phi_{m,\alpha,\beta,\bfk}(\bfx)$, and the Berry connection takes the form $a_1(\bfk) = 0$ and $a_2(\bfk)=  k_1/B$ for all Landau levels. 
        That is, it describes states with a homogeneous Berry flux, and the corresponding Chern number of the different Landau levels is $\Cnumber=1$. This property is shown in Lemma~\ref{lem: berry_connection}.   
        
        In the presence of a periodic external potential, the eigenstates $\psi_{n,\bfk}$ can be formally expanded in the basis of the magnetic translation eigenfunctions, such that     
        \begin{equation}
                \psi_{n,\bfk}(z,\conjz) = \sum_{m, \alpha} c_{n,m,\alpha,\beta}(\bfk)  \phi_{m,\alpha,\beta, \bfk}(z,\conjz). \label{expansion}
        \end{equation}
        Then, the Chern number reduces to
        \begin{equation}
                \Cnumber_n = 1 + \frac{1}{2 \pi } \int_{\Omega^*} \left[ \partial_{k_1} \tilde{a}_{n,2}(\bfk) -\partial_{k_2} \tilde{a}_{n,1}(\bfk)   \right] \,\mathrm{d}\bfk \equiv 1+ \tilde{\Cnumber}_n, \label{chern-number}
        \end{equation}
        with 
        \begin{equation}
                \tilde{a}_{n,\nu}(\bfk) = - \iu \sum_{m,\alpha} \overline{c}_{n,m,\alpha,\beta}(\bfk) \partial_{k_\nu} c_{n,m,\alpha,\beta}(\bfk),\quad\nu=1,2. 
        \end{equation}
        This property follows again from Lemma~\ref{lem: berry_connection} and the expansion in Eq.~\eqref{expansion}:
        the Berry connection is split into three parts
        \begin{align}
                a_{n,1}(\bfk) &=   \tilde{a}_{n,1}(\bfk) + \Delta   a_{n,1}(\bfk) \\
                a_{n,2}(\bfk) &= \frac{k_1}{B}   +  \tilde{a}_{n,2}(\bfk) +\Delta  a_{n,2}(\bfk), 
        \end{align}
        where the last term $\Delta  a_{n,\nu}(\bfk) $ follows from the coupling between different Landau levels. 
        For example, it takes the form

        \begin{align}
            &\Delta  a_{n,1}(\bfk)\\ = &\sum_{\alpha, \alpha',m, m'} \overline{c}_{n,m,\alpha,\beta}(\bfk)   \frac{ 1}{\sqrt{2B}} \delta_{\alpha,\alpha'} \left( \sqrt{m}\delta_{m,m'+1}+\sqrt{m'}\delta_{m+1,m'} \right)   c_{n,m',\alpha',\beta}(\bfk) 
        \end{align}
        and analogously for $\Delta  a_{n,2}(\bfk) $. Note that these terms are continuous, periodic, and invariant under gauge transformations 
        $c_{n,m,\alpha,\beta}({\bfk}) \rightarrow e^{\iu \theta_{\bf k}} c_{n,m,\alpha,\beta}({\bfk})$ of the eigenstates $\psi_{n,{\bfk}}$, and therefore do not contribute to the Chern number $\Cnumber_n$. Then, Eq.~\eqref{chern-number} immediately follows from the remaining terms. 
            
        Therefore, the Chern number can be effectively determined using the Fukui–Hatsugai–Suzuki method \cite{fukuiChernNumbersDiscretized2005} from the numerical determination of the expansion coefficients $c_{n,m,\alpha,\beta}(\bfk)$ 
        over a finite grid of $\bfk$-values in the Brillouin zone. In the presence of degeneracies, which appear when the flux is such that $\beta>1$, the Chern number of a band can be computed by considering all the eigenpairs associated with this band, and the application of the Fukui–Hatsugai–Suzuki method for such degenerate bands.
        
        \begin{lemma}\label{lem: berry_connection}
            The functions $u_{m,\alpha,\beta, \bfk}(\bfx)=\exp(-\iu\bfx\cdot\bfk)\phi_{m,\alpha,\beta,\bfk}(\bfx)$ constructed from the basis states satisfy
            \begin{align}
                \langle u_{m,\alpha,\beta,\bfk} | \partial_{k_1} u_{m',\alpha',\beta,\bfk}\rangle & = - \frac{ \iu}{\sqrt{2B}} \delta_{\alpha,\alpha'} \left( \sqrt{m}\delta_{m,m'+1}+\sqrt{m'}\delta_{m+1,m'} \right),\\
                \langle u_{m,\alpha,\beta,\bfk} | \partial_{k_2} u_{m',\alpha',\beta,\bfk}\rangle & = \frac{ \iu  k_1}{B} \delta_{m,m'} \delta_{\alpha,\alpha'} +  \frac{1}{\sqrt{2B}} \delta_{\alpha,\alpha'}  \left( \sqrt{m}\delta_{m,m'+1}-\sqrt{m'}\delta_{m+1,m'} \right).
            \end{align} 
        \end{lemma}
        
        \begin{proof} First, we evaluate $\langle u_{0,\alpha,\beta,\bfk} | \partial_{k_i} u_{0,\alpha',\beta,\bfk}\rangle$ and follow the proof of Lemma~\ref{lem: Normalization Constant},
            \begin{align}
                &\int_{\Omega} \overline{ u}_{0,\alpha,\beta,\bfk}(\bfx) \partial_{k_1} u_{0,\alpha',\beta,\bfk}(\bfx) \,\mathrm{d}\bfx\\ = & N_{k_1}^2 \int_{\Omega} e^{-Bx_2^2-2k_1 x_2}  \sum_{n,n'\in \bbZ}  \left(  -x_2- \gamma_n(\alpha )- \frac{k_1}{B}\right) \nonumber\\  
                &\quad \times C_{n,n'}  F_{n,n'}(x_1) G_{n,n'}(x_2) H_{n,n'}(\bfk) \,\mathrm{d}\bfx \\
                = & -  \delta_{\alpha, \alpha'}    N_{k_1}^2 \int_\bbR  \left(u + \frac{\alpha}{\beta}+\frac{k_1}{B} \right) \exp\left(-B\left(u + \frac{\alpha}{\beta}+\frac{k_1}{B}\right)^2\right)\,\mathrm{d}u  \nonumber\\
                = & 0.
            \end{align}
            On the other hand, we obtain
            \begin{align}
                &\int_{\Omega} \overline{ u}_{0,\alpha,\beta,\bfk}(\bfx) \partial_{k_2} u_{0,\alpha',\beta,\bfk}(\bfx) \,\mathrm{d}\bfx\\ = & N_{k_1}^2 \int_{\Omega} e^{-Bx_2^2-2k_1 x_2}  \sum_{n,n'\in \bbZ}  \left(  - \iu x_2- \frac{\iu\gamma_n(\alpha )}{B}\right) \nonumber\\  
                &\quad \times C_{n,n'}  F_{n,n'}(x_1) G_{n,n'}(x_2) H_{n,n'}(\bfk) \,\mathrm{d}\bfx \\
                = & -  \delta_{\alpha, \alpha'}   \iu N_{k_1}^2 \int_\bbR  \left(u + \frac{\alpha}{\beta} \right) \exp\left(-B\left(u + \frac{\alpha}{\beta}+\frac{k_1}{B}\right)^2\right)\,\mathrm{d}u \nonumber\\
                = & \frac{\iu k_1}{B}  \delta_{\alpha, \alpha'}.   
            \end{align}
            For the overlap for higher Landau levels, we use that the normalization $\langle u_{m,\alpha,\beta,\bfk} | u_{m',\alpha',\beta,\bfk}\rangle = \delta_{m,m'}\delta_{\alpha, \alpha'}$ implies
            \begin{equation}
                \langle u_{m,\alpha,\beta,\bfk} | \partial_{k_i} u_{m',\alpha',\beta,\bfk}\rangle = - \overline{ \langle u_{m',\alpha',\beta,\bfk} | \partial_{k_i} u_{m,\alpha,\beta,\bfk}\rangle}.
            \end{equation}
            Therefore, we can assume without loss of generality that $m\geq m'$. Using the commutation relation
            \begin{equation} 
                \left[ \left(L^{-}\right)^m, e^{\iu \bfk \cdot \bfx} \partial_{k_i}e^{-\iu \bfk \cdot \bfx}  \right] = - m \iu \sqrt{\frac{1}{2B}} \left(\partial_{k_i} (k_1 + \iu k_2)\right) \left(L^{-}\right)^{m-1},
            \end{equation}
            as well as the standard relations of the raising and lowering operators $L^{\pm}$ applied to the ground state $|u_{0,\alpha,\beta,\bfk}\rangle$, the equations in the lemma immediately follow. 
                
        \end{proof}

%% file: Numerics.tex
\section{Methodology\label{sec:Numerics}}
    In this section, we explain in detail how to obtain numerical solutions to Eq.~\eqref{eq: Eigenvalue Problem at k}. Given a fixed $\bfk\in\Omega^*$, we will consider a finite number of normalized eigenfunctions of the kinetic energy operator $L(\bfA)$, which we denote as $\set{\phi_{m,\bfk,\alpha,\beta}}_{m=0}^N$, project the operator $H(\bfA)=L(\bfA)+V$ onto this subspace, and solve the resulting finite-dimensional eigenvalue problem.
    
    \subsection{The Galerkin Projection Method}
        Let us quickly review the Galerkin projection method, also called the Rayleigh-Ritz method. Suppose we want to compute the eigenvalues and eigenvectors of a self-adjoint operator $H$ on a complex Hilbert space $\mcX$. We consider a finite-dimensional subspace $\mcV\subset \mcX$ with a (not necessarily orthonormal) basis $\set{\phi_m}_{m=0}^N$. Our goal is to approximate the eigenpair $(\Psi,\lambda) \in\mcX\times\bbR$, which satisfies the eigenvalue equation
        \begin{equation}
             H\Psi=\lambda\Psi,
        \end{equation}
        in $\mcX$. We first restrict $\Psi$ to $\mcV$; in other words, we express $\Psi$ as a linear combination of the chosen basis functions for $\mcV$:
        \begin{equation} 
            \Psi =\sum_{m=0}^N c_m\phi_m,\quad c_m\in \bbC,
        \end{equation}
        substitute it into the eigenvalue equation, and take the inner product against $\phi_\ell$ for $\ell=0,\ldots, N$ to get
        \begin{equation}
            \sum_{m=0}^N c_m\inpro{\phi_\ell}{H\phi_m}=\lambda \sum_{m=0}^Nc_m\inpro{\phi_\ell}{\phi_m},
        \end{equation}     
        which can be written in a more compact matrix form as
        \begin{align}
            \bfH_\mcV \bfc_\mcV&=\lambda{\bfS}_{\mcV} \bfc_\mcV,\\
            (\bfH_{\mcV})_{\ell m} &= \inpro{\phi_\ell}{H\phi_m},\\
            (\bfS_{\mcV})_{\ell m}&= \inpro{\phi_\ell}{\phi_m},\qquad \ell,m=0,\ldots, N,\\
            (\bfc_{\mcV})_\ell &=c_\ell.
        \end{align}
        What we have accomplished through this procedure is reducing an eigenvalue problem involving a linear differential operator to a matrix eigenvalue problem, at the price of obtaining only an approximate solution. However, from a computational perspective, there are still some significant challenges when applying the previous procedure to $H=L(\bfA)+V$: (i) the numerical construction of $\phi_m$ itself for large values of $m$, and (ii) the computation of the matrix elements $(H_{\mcV})_{\ell m}$, where $\ell,m=0,\ldots,N$, in particular the cross-terms between the basis functions and the potential, i.e., $\inpro{\phi_\ell}{V\phi_m}$ for $\ell,m=0,\ldots,N$. We now propose strategies to deal with these issues.
    
    \subsection{\label{subsec: Construction of the Basis} Construction of the Basis}
        Let us tackle the first complication described above, that is, the construction of an adequate basis for the approximation subspace $\mcV$. Inspired by the plane-wave methods used in density functional theory~\cite{cancesNumericalAnalysisPlanewave2012}, our strategy consists of constructing a basis of eigenfunctions of the Landau operator $L(\bfA)$ that satisfy the magnetic boundary condition in Eq.~\eqref{eq:Magnetic boundary condition}. We denote this set by 
        \begin{equation} 
            \set{\phi_{m,\bfk,\alpha,\beta} : m=0,\dots, N, \, \bfk=(k_1,k_2)\in\Omega^*, \, \alpha = 0,\dots,\beta-1 }, 
        \end{equation} 
        though for the sake of clarity, we will drop the dependence on $\bfk$, $\alpha$, and $\beta$ in the notation whenever possible.
        
        The construction of $\phi_m$ involves the iterated application of the raising operator $L^+$ to the ground-state eigenfunction $\phi_0$. By manipulating the terms in $\phi_0$, we obtain a closed expression for $\phi_m$, which can be directly implemented in a computer program. As before, it is convenient to work in complex coordinates. First, since the differential operators $\partial_{z}$ and $\partial_{\conjz}$ commute with the multiplication operators $\conjz$ and $z$, respectively, it is possible to expand the raising operator $(L^+)^m$ for any $m\in\bbN_0$. The resulting identity is 
        \begin{align}
            (L^+)^m &= 
            \begin{cases}
                1, &m=0,\\
                \displaystyle\sum_{i=0}^m (C_1)^iC_2^{m-i}(\partial_{z})^i \conjz^{m-i}, & m\geq 1,
            \end{cases}\label{eq: Raising Operator Binomial Expansion}\\
            C_1&=-\sqrt{\frac{2}{B}},\quad C_2=\frac{1}{2}\sqrt{\frac{B}{2}}.
        \end{align}
        Next, we write the ground-state function $\phi_{0}$ as follows:
        \begin{align}
            \phi_{0}(z,\conjz)&=N_{k_1}\exp\left(\frac{B}{4}z^2\right)\sum_{n\in\bbZ}f_{n}\exp(h_{n}(\conjz)z),\\
            f_{n}&=\exp\left(-\pi\beta(n+\alpha/\beta)^2-2\pi \iu \beta(n+\alpha/\beta)\delta(\bfk)\right),\\
            h_{n}(\conjz) &= -\frac{B}{4}\conjz+\iu k_1+2\pi \iu\beta(n+\alpha/\beta).
        \end{align}
        We then apply $(L^+)^m$ in the form presented in Eq.~\eqref{eq: Raising Operator Binomial Expansion} and apply the product rule, which leads to the following explicit expression for $\phi_{m}$: 
        \begin{align}
            \phi_{m}(z,\conjz)&=N_{k_1}\sum_{i=0}^m \binom{m}{i}C_1^i C_2^{m-i}\conjz^{m-i}\sum_{j=0}^i \binom{i}{j}G_j(z)\sum_{n\in\bbZ}f_{n}H_{n,i-j}(z,\conjz)\label{eq: Closed Form of the Basis Functions},\\
            G_{j}(z)&=P_j(z)\exp\left(\frac{B}{4}z^2\right),\\
            P_j(z)&=\begin{cases}
            1 & j=0,\\
            P_{j-1}'(z)+\frac{B}{2}P_{j-1}(z)z & j>0,
            \end{cases}\\
            H_{n,i-j}(z,\conjz)&=(h_{n}(\conjz))^{i-j}\exp(h_{n}(\conjz)z),
        \end{align}
        which holds for any $m\in\bbN_0$. 
       
        \subsection{\label{subsec: Computation of the Matrix Elements}Computation of the Matrix Elements}
        Now we focus on the computation of the matrix elements, that is, terms of the form $\inpro{\phi_\ell}{H\phi_m}$. To begin,
        \begin{equation}
            \inpro{\phi_\ell}{H\phi_m}=\inpro{\phi_\ell}{L(\bfA)\phi_m}+\inpro{\phi_\ell}{V\phi_m}\label{eq: Matrix Elements of the Hamiltonian, Linearity}.
        \end{equation}
        The first term on the right-hand side of Eq.~\eqref{eq: Matrix Elements of the Hamiltonian, Linearity} is easy to evaluate because $\phi_m$ is an eigenfunction of $L(\bfA)$, and the basis functions are orthonormalized, so
        \begin{equation}
            \inpro{\phi_\ell}{L(\bfA)\phi_m}=B\left(\ell+\frac{1}{2}\right)\delta_{\ell m}\label{eq: kinetic energy for matrix elements final form}.
        \end{equation}
        The second term on the right-hand side of Eq.~\eqref{eq: Matrix Elements of the Hamiltonian, Linearity}, $\inpro{\phi_\ell}{V\phi_m}$ for $0\leq \ell,m\leq N$, is more challenging. Even with the closed expression for the basis functions, Eq.~\eqref{eq: Closed Form of the Basis Functions}, directly computing the integrals is infeasible, as is attempting to use numerical methods without prior simplification. In order to simplify the expressions, we exploit the property $(L^+)^\dagger = L^-$. 
        
        To proceed, let us distinguish two cases: $\ell\leq m$ and $\ell>m$. If we obtain an expression for the former, the latter can be easily computed from the relationship
        \begin{equation}
            \inpro{\phi_\ell}{V\phi_m}=\overline{\inpro{\phi_m}{\overline{V}\phi_\ell}}=\overline{\inpro{\phi_m}{V\phi_\ell}},
        \end{equation}
        where the last equality holds for real-valued potentials. 
    
    \subsubsection{\texorpdfstring{The Case $\ell\leq m$}{The Case l <= m}}
        We have the following product rule for the lowering operator, which is a consequence of the product rule for derivatives and the commutation properties between $\partial_{\conjz}$ and $z$:
        \begin{equation}
            (L^-)^\ell(F(z,\conjz)G(z,\conjz)) = \sum_{j=0}^\ell\binom{\ell}{j}(C_1\partial_{\conjz})^j F(z,\conjz) (L^-)^{\ell-j}G(z,\conjz).
        \end{equation}
        Now, using this with $F(z,\conjz)=V(z,\conjz)$ and $G(z,\conjz)=\phi_{m}(z,\conjz)$ yields
        \begin{align}
            \inpro{\phi_\ell}{V\phi_m}&=\frac{1}{\sqrt{\ell!}}\inpro{\phi_0}{(L^-)^\ell (V\phi_m)}\\
            &=\sqrt{\frac{m!}{\ell!}}\sum_{j=0}^\ell\binom{\ell}{j}\frac{1}{\sqrt{(m-\ell+j)!}}C_1^j\inpro{\phi_0}{(\partial_{\conjz}^j V)\phi_{m-\ell+j}}\label{eq: Potential Cross Terms}.
        \end{align}
        In some cases, such as when the potential consists of additive sinusoidal terms, the computation of $\partial_{\conjz}^j V$ is straightforward. On the other hand, it is possible in full generality that the potential is not simple enough to allow closed expressions for its derivatives or is simply not regular enough. In that case, we can exploit its periodicity. By using its Fourier representation, 
        \begin{align}
            V(\bfx) &=\sum_{\bfxi\in\mcL^*}\hat{V}_\bfxi\exp(\iu\bfx\cdot \bfxi),\\
            \mcL^*&=\set{(\xi_{n_x},\xi_{n_y})=2\pi(n_x,n_y): n_x,n_y\in\bbZ},
        \end{align}
        we can easily compute $\partial_{\conjz}^jV$. In fact, after returning to the spatial coordinates $\bfx=(x_1,x_2)$,
        \begin{equation}
            (\partial_{\conjz}^j V)(x_1,x_2)=\frac{1}{2^j}\sum_{\bfxi\in\mcL^*} \hat V_\bfxi\sum_{k=0}^j\binom{j}{k}\iu^{2j-k}\xi_{n_x}^k\xi_{n_y}^{j-k}\exp(\iu\bfx\cdot\bfxi).
        \end{equation} 
        If we substitute this back into Eq.~\eqref{eq: Potential Cross Terms}, we get
        \begin{align}
        \inpro{\phi_\ell}{V\phi_m} &= \sqrt{\frac{m!}{\ell!}} \sum_{j=0}^\ell \binom{\ell}{j}\frac{1}{\sqrt{(m-\ell+j)!}} \left(\frac{C_1}{2}\right)^j \sum_{\bfxi\in{\mcL^*}} \hat V_{\bfxi} \nonumber \\
        &\quad \times \sum_{k=0}^j \binom{j}{k} \iu^{2j-k}\xi_{n_x}^k \xi_{n_y}^{j-k} \inpro{\phi_0}{\exp(\iu\bfx\cdot\bfxi)\phi_{m-\ell+j}}.
        \end{align}
        This expression does not seem too promising, at least at first sight. However, the terms of the form 
        \begin{equation}
            \inpro{\phi_0}{ \exp(\iu\bfx\cdot \bfxi) \phi_{m-\ell+j}}
        \end{equation}
        can be further simplified. For clarity, we make the dependence on $\alpha$ explicit in the following result.
        
        \begin{lemma}
            Let $\bfk=(k_1,k_2)\in\Omega^*$. For any $m\in\bbN_0$, and any $\alpha,\alpha'\in\set{0,\ldots,\beta-1}$, we have
            \begin{equation}
                \inpro{\phi_{0,\alpha}}{\exp(\iu\bfx\cdot \bfxi) \phi_{m,\alpha'}}=\frac{1}{\sqrt{m!}}\left(\frac{C_1}{2}(\iu k_1+k_2)\right)^m\inpro{\phi_{0,\alpha}}{\exp(\iu \bfx\cdot \bfxi)\phi_{0,\alpha'}}\label{eq: Fourier Tranform fmn}.
            \end{equation}
        \end{lemma}
        
        \begin{proof}
            We proceed by induction. Since $L^-\phi_{0,\alpha}=0$,
            \begin{align}
                L^-(\exp(-\iu\bfx\cdot \bfxi)\phi_{0,\alpha}) &= \exp(-\iu\bfx\cdot \bfxi) L^-\phi_{0,\alpha} + C_1(\partial_{\conjz}\exp(-\iu\bfx\cdot \bfxi))\phi_{0,\alpha}\\
                &=\frac{C_1}{2}(-\iu\xi_{n_x}+ \xi_{n_y})\exp(-\iu\bfx\cdot \bfxi)\phi_{0,\alpha}\label{eq: Fourier Transform phi exp}.
            \end{align}
            Then the base case with $m=1$ is immediate. Now assume the identity in Eq.~\eqref{eq: Fourier Tranform fmn} holds for $m\in\bbN_0$. Then 
            \begin{align}
                 &\inpro{\phi_{0,\alpha}}{\exp(\iu\bfx\cdot \bfxi) \phi_{m+1,\alpha'}}\\
                 = &\frac{1}{\sqrt{m+1}}\inpro{L^-(\phi_{0,\alpha}\exp(-\iu\bfx\cdot\bfxi))}{\phi_{m,\alpha'}}\\
                 = &\frac{C_1}{2\sqrt{m+1}}(\iu \xi_{n_x}+\xi_{n_y})\inpro{\phi_{0,\alpha}}{\exp(\iu\bfx\cdot \bfxi)\phi_{m,\alpha'}}\\
                 = &\frac{1}{\sqrt{(m+1)!}}\left(\frac{C_1}{2}(\iu k_1+k_2)\right)^{m+1}\inpro{\phi_{0,\alpha}}{\exp(\iu \bfx\cdot \bfxi)\phi_{0,\alpha'}},
            \end{align}
            where we have used the induction hypothesis in the last equality.
        \end{proof}
        
       The advantage of Eq.~\eqref{eq: Fourier Tranform fmn} is that it allows us to compute the terms $\inpro{\phi_{0,\alpha}}{\exp(\iu\bfx\cdot \bfxi)\phi_{m-\ell+j,\alpha'}}$ by just using $\inpro{\phi_{0,\alpha}}{\exp(\iu \bfx\cdot \bfxi)\phi_{0,\alpha'}}$. In particular, it replaces the multiple sums appearing in Eq.~\eqref{eq: Closed Form of the Basis Functions}, which originate from the derivatives, with a single polynomial expression in terms of $k_1$ and $k_2$. Not only is this computationally more efficient, but by lowering the number of sums, the chance of error propagation is also reduced. The final expression for the computation of the terms involving the potential is: 
        \begin{align}
        &\inpro{\phi_{\ell,\alpha}}{V\phi_{m,\alpha'}} \nonumber \\ 
        = & \sqrt{\frac{m!}{\ell!}} \sum_{j=0}^\ell \binom{\ell}{j}\frac{1}{\sqrt{(m-\ell+j)!}} \left(\frac{C_1}{2}\right)^j \sum_{\bfxi\in{\mcL^*}} \hat V_{\bfxi} \\
        &\quad \times \sum_{k=0}^j \binom{j}{k} \frac{\iu^{2j-k}\xi_{n_x}^k \xi_{n_y}^{j-k}}{\sqrt{(m-\ell+j)!}} \left(\frac{C_1}{2}(\iu k_1+k_2)\right)^{m-\ell+j}\inpro{\phi_{0,\alpha}}{\exp(\iu\bfx\cdot\bfxi)\phi_{0,\alpha'}} \nonumber \\
        = &\sqrt{\frac{m!}{\ell!}} \sum_{j=0}^\ell \binom{\ell}{j}\frac{1}{(m-\ell+j)!} \left(\frac{C_1}{2}\right)^j   
        \sum_{k=0}^j \binom{j}{k} \left(\frac{C_1}{2}(\iu k_1+k_2)\right)^{m-\ell+j}  \\
        &\quad \times \sum_{\bfxi\in{\mcL^*}} \hat V_{\bfxi} \iu^{2j-k}\xi_{n_x}^k \xi_{n_y}^{j-k} \inpro{\phi_{0,\alpha}}{\exp(\iu\bfx\cdot\bfxi)\phi_{0,\alpha'}}\label{eq: potential for matrix elements final form}.
       \end{align}

%% file: Results.tex
\section{Results and Discussion\label{sec: Results}}
In this section, we present numerical experiments to test the proposed Galerkin method described previously. The algorithms were implemented in Python and are fully available. 

\subsection{Integer Fluxes}
We tested the method using a well-known periodic sinusoidal potential of the form:
\begin{equation} 
    V(\bfx) = \frac{V_0}{4} \left[ 2 - \cos\left(\frac{2\pi}{a}x_1\right) - \cos\left(\frac{2\pi}{a}x_2\right) \right]\label{eq: potential for simulations}, 
\end{equation} 
where $a$ is the lattice constant. The system is subject to a perpendicular uniform magnetic field of strength $B$. In our experiments, we fixed $a=1$, but the code is implemented to allow for variable $a$. The parameter $V_0$ modulates the relative strength of the potential: high values of $V_0$ in relation to the magnetic field strength $|B|$ correspond to the regime governed by the deep wells of the potential (yielding an optical lattice), while small values recover the weak-potential limit discussed in the introduction. We numerically constructed the band structure for the different regimes described below and computed the associated Chern numbers using the Fukui-Hatsugai-Suzuki method~\cite{fukuiChernNumbersDiscretized2005}.
    
We first performed a numerical study on two distinct systems: one with non-degenerate Landau levels ($\beta=1$) and one where degeneracy is present ($\beta=4$), corresponding to $B=2\pi$ and $B=8\pi$, respectively. The simulation parameters are summarized in Table~\ref{tab: simulation parameters}, which lists the physical parameters, their variable names in the code, and their specific values. To test the method across different physical regimes, we varied the potential amplitude $V_0$ in Eq.~\eqref{eq: potential for simulations} while keeping the magnetic field strength constant. The results are presented as a function of the ratio between these parameters, given by:
\begin{equation}
    R = \frac{V_0}{|B|}.
\end{equation}
The specific values of $R$ used in the simulations, alongside the corresponding Chern numbers, are listed in Table~\ref{tab: potential data}. For the Chern number computation, we employed a regular grid on the Brillouin zone defined as follows:
\begin{equation} 
    \mathcal{K} = \left\{ \left( -\pi + \frac{2\pi i}{n_{k_1}-1}, \, -\pi + \frac{2\pi j}{n_{k_2}-1} \right) : i \in \{0, \dots, n_{k_1}-1\}, \, j \in \{0, \dots, n_{k_2}-1\} \right\}, 
\end{equation} 
where $n_{k_1}$ and $n_{k_2}$ denote the number of grid points in the $k_1$ and $k_2$ directions, respectively. For our simulations, we set $n_{k_1} = n_{k_2} = 11$. When constructing the band structure, we followed the path $M\rightarrow X\rightarrow \Gamma\rightarrow M$ and employed a higher number of $k$-points per edge (\texttt{points\_per\_segment}), as detailed in Table~\ref{tab: simulation parameters}. 
    
Finally, two different parameters are used to control the number of basis functions employed in the simulations. The number of basis functions used in the band structure computations is controlled by the variable \texttt{N\_bands}, while the number used for the Chern number computation is \texttt{N\_topo}. In the former, the number of basis functions is $(\texttt{N\_bands})\beta$, while in the latter it is $(\texttt{N\_topo})\beta$, and in general $\texttt{N\_topo}<\texttt{N\_bands}$. We chose different values because the band structure can be reconstructed by requiring only the eigenvalues of the system; as explained in Section~\ref{subsec: Computation of the Matrix Elements}, the matrix elements can be reduced to inner products depending only on the ground-state eigenfunctions $\set{\phi_{0, \bfk, \alpha,\beta}: \bfk\in\Omega^*, \alpha=0,\ldots,\beta-1}$. In particular, this allows us to work with higher excitation bands. Computing the Chern number, in contrast, explicitly requires access to all the basis functions $\set{\phi_{m, \bfk, \alpha,\beta}: m=0,\ldots,\texttt{N\_topo}, \bfk\in\Omega^*, \alpha=0,\ldots,\beta-1}$. Since the construction of the basis is computationally expensive, we are forced to work with a smaller number of basis functions for the calculation of the Chern number.
    
\begin{table}[htbp]
\centering
\caption{Overview of common and variable simulation parameters.}
\label{tab: simulation parameters}
    \begin{tabular}{llc}
        \toprule
        \textbf{Parameter Name} & \textbf{Variable} & \textbf{Value} \\ 
        \midrule
        
        \multicolumn{3}{l}{\textbf{Common Parameters}} \\
        \midrule
        Lattice Constant & \texttt{a} & $1$ \\
        $\bfk$-points & \texttt{nkx}$\times$\texttt{nky}  & $11\times 11$ \\
        $\Theta$ function truncation & \texttt{trunc} & $10$ \\
        Fourier series truncation & \texttt{Ecut} & $2^6$ \\
        Points per segment & \texttt{points\_per\_segment} & $100$ \\
        Number of States for Bands & \texttt{N\_bands} & $50$ \\
        Number of States for Chern number & \texttt{N\_topo} & $30$ \\
        \midrule

        \multicolumn{3}{l}{\textbf{Variable Parameter: Magnetic Field Strength}} \\
        \midrule
        Non-degenerate case & \texttt{B} & $2\pi$ \\
        Degenerate case & \texttt{B} & $8\pi$ \\
        \bottomrule
        Rational flux & \texttt{B} & $6\pi$ \\
        \bottomrule
    \end{tabular}
\end{table}
    
\begin{table}[htbp]
\centering
\caption{Various ratios of potential amplitude to magnetic field strength ($R$) and the corresponding Chern numbers. The middle column shows the results for a system using a non-degenerate basis ($\beta=1$). The right column shows the Chern number for a system using a degenerate basis ($\beta=4$). Blank cells (--) indicate the ratio was not evaluated for that specific configuration.}
        \label{tab: potential data}
        \begin{tabular}{ccc}
            \toprule
            & \multicolumn{2}{c}{$\Cnumber$} \\ 
            \cmidrule(lr){2-3} 
            $R$ & $\beta = 1$ & $\beta = 4$ \\
        \midrule
        $10^{-5}$ & $1$  & $1$ \\
        $0.1$     & $1$  & $1$ \\
        $1.0$     & $1$  & $1$ \\
        $2.0$     & $1$  & $1$ \\
        $3.0$     & $1$  & $1$ \\
        $3.1$     & $1$  & --  \\
        $3.16833$ & $1$  & --  \\
        $3.2$     & $0$  & --  \\
        $3.8$     & --   & $1$ \\
        $3.90355$ & --   & $0$ \\
        $4.0$     & --   & $0$ \\
        $5.0$     & $0$  & --  \\
        $10.0$    & $0$  & $0$ \\
        $20.0$    & $0$  & $0$ \\
        \bottomrule
    \end{tabular}
\end{table}
    
To obtain a better estimate of the potential strength at the band crossing, we applied a simple bisection procedure on the ratio $R$. The procedure is as follows:
\begin{enumerate}
    \item \textbf{Initialization:} We identified an initial interval $[R_{l,0}, R_{u,0}]$ such that the corresponding Chern numbers are $\mathcal{C}_l = 1$ and $\mathcal{C}_u = 0$, respectively.
    \item \textbf{Iteration:} At each step $i$, we computed the midpoint $R_{m,i} = \frac{1}{2}(R_{l,i} + R_{u,i})$ and evaluated its Chern number $\mathcal{C}_m$.
    \begin{itemize}
        \item If $\mathcal{C}_m = 1$, we updated the lower bound: $R_{l,i+1} = R_{m,i}$ and $R_{u,i+1} = R_{u,i}$.
        \item Otherwise (if $\mathcal{C}_m = 0$), we updated the upper bound: $R_{u,i+1} = R_{m,i}$ and $R_{l,i+1} = R_{l,i}$.
    \end{itemize}
    \item \textbf{Termination:} We repeated this process until the interval width $R_{u,n} - R_{l,n} < \texttt{TOL}$.
\end{enumerate}
The required number of iterations $n$ can be estimated using the formula:
\begin{equation} 
    n > \frac{\log\left( \frac{R_{u,0} - R_{l,0}}{\texttt{TOL}} \right)}{\log(2)}.
\end{equation}

We now discuss the simulation results:
\begin{enumerate}[label=(\roman*)] 
    \item \emph{Non-degenerate case ($B=2\pi$)}: In this regime, $\beta=1$, indicating non-degenerate Landau levels. In Figure~\ref{fig: non_degenerate_bands}, we plot selected band structures, while the corresponding Chern numbers are depicted in Figure~\ref{fig: non_degenerate_chern}. In the weak potential regime (Figures~\ref{fig: non_degenerate_bands_0_0100} and~\ref{fig: non_degenerate_bands_0_1000}), the bands are nearly flat, closely resembling the unperturbed Landau levels, and the computed Chern number evaluates exactly to $1$ within machine precision. As we increase the potential strength (Figures~\ref{fig: non_degenerate_bands_3_0000} and~\ref{fig: non_degenerate_bands_3_2000}), we observe band deformation; specifically, the lowest and first excited bands begin approaching each other, indicating that the band gap is closing. In fact, in Figure~\ref{fig: non_degenerate_bands_3_2000}, the bands appear to touch. Although numerical approximations and the limitations of our discretized scheme prevent us from definitively confirming an exact gap closure at $V_0/|B|=3.2$, the computed Chern number drops to $0$. This transition indicates that a crossover indeed occurs between $V_0/|B|=3$ and $V_0/|B|=3.2$. To obtain a more precise estimate, we performed the bisection method with $R_l=3.0$ and $R_u=3.2$. The results are shown in Table~\ref{tab:bisection_non_degenerate}, which pinpoints the transition at $R\approx 3.16833$. The eigensurface of the system before and after the band crossing is shown in Figure~\ref{fig: non_degenerate_surface}. As the potential increases further (Figures~\ref{fig: non_degenerate_bands_20_0000}), the bands begin to flatten again. This is consistent with the system entering an optical lattice regime governed by deep potential wells. Figure~\ref{fig: non_degenerate_density} displays the ground-state density transitions. At moderate potentials, the density begins to localize near the potential wells, and by $V_0/|B|=20$, the localization of the electrons at the potential wells becomes highly evident across the supercell.  

    \item \emph{Degenerate case ($B=8\pi$)}: In this regime, $\beta=4$, leading to a four-fold degeneracy of the Landau levels. We chose the potential strengths such that most of the tested ratios $R$ match those used in the non-degenerate case. The evaluated ratios differ between the two systems only near their respective band crossing points. The resulting band structures are shown in Figure~\ref{fig: split_bands}, and the corresponding Chern numbers are depicted in Figure~\ref{fig: split_chern}. Upon introducing the potential, the degeneracy is lifted and the bands split, a phenomenon evident even in the weak potential regime (Figures~\ref{fig: split_bands_0_0100} and~\ref{fig: split_bands_0_1000}). Note that the band splitting is not instantaneous: some bands still touch at the $\Gamma$ point. As the potential strength increases, the lowest bands begin to flatten, while the third and fourth bands approach each other (Figures~\ref{fig: split_bands_3_8000} to~\ref{fig: split_bands_4_0000}). Between the ratios $V_0/|B|=3.8$ and $V_0/|B|=4.0$, these bands cross, closing the gap. As before, we employed the bisection method to estimate the transition value. The results are shown in Table~\ref{tab:bisection_split}, identifying the transition at $R\approx 3.90355$. As expected, this gap closure is accompanied by the total Chern number jumping from $1$ to $0$. Further increasing the potential strength (Figures~\ref{fig: split_bands_4_0000} to~\ref{fig: split_bands_20_0000}) leads to completely flat bands, as the system once again behaves like an optical lattice. The corresponding ground-state densities (not depicted here) for the weak and strong potential regimes are similar to those observed in the non-degenerate case; that is, under the influence of the potential, the electron density begins localizing at the potential wells. 
\end{enumerate}

\begin{figure}[htbp]
\centering   
    \begin{subfigure}[t]{0.45\textwidth}
        \centering
        \includegraphics[width=\linewidth]{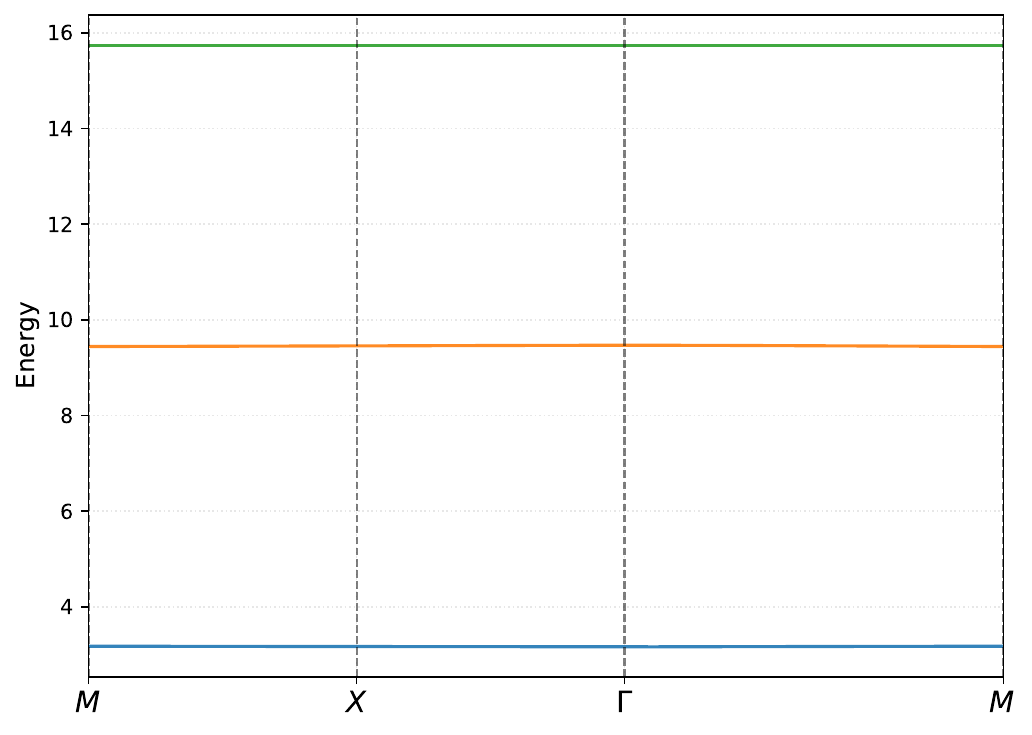}
        \caption{$V_0/|B|=0.01$}
        \label{fig: non_degenerate_bands_0_0100} 
    \end{subfigure}\hfill 
    \begin{subfigure}[t]{0.45\textwidth}
        \centering
        \includegraphics[width=\linewidth]{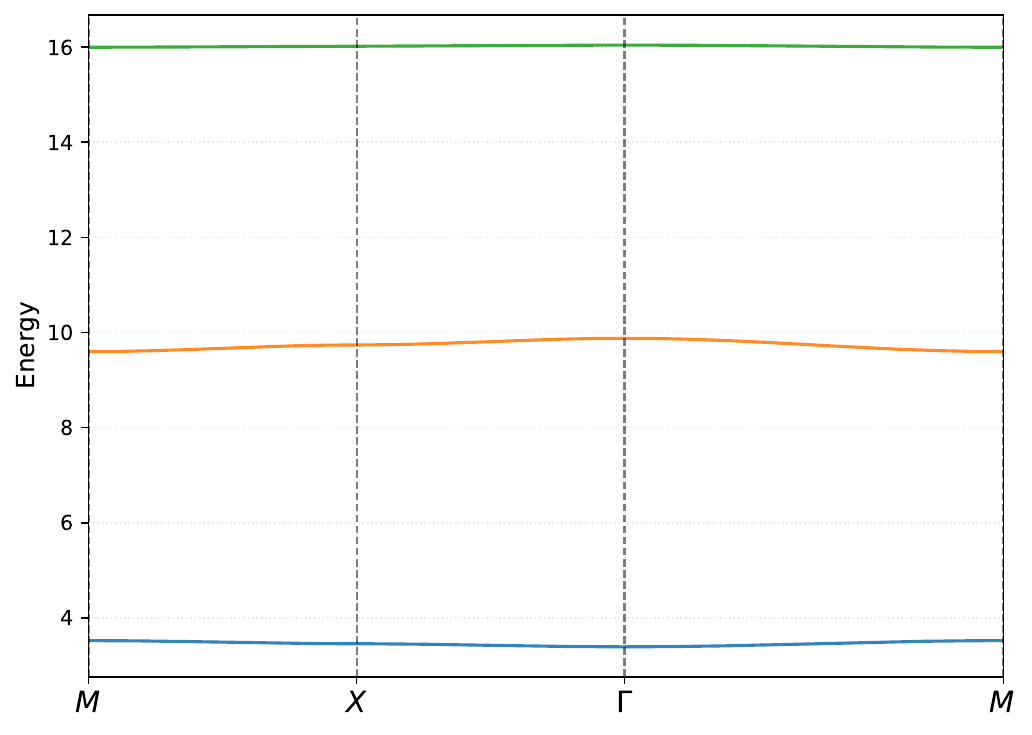}
        \caption{$V_0/|B|=0.1$}
        \label{fig: non_degenerate_bands_0_1000} 
    \end{subfigure}

    \begin{subfigure}[t]{0.48\textwidth}
        \centering
        \includegraphics[width=\linewidth]{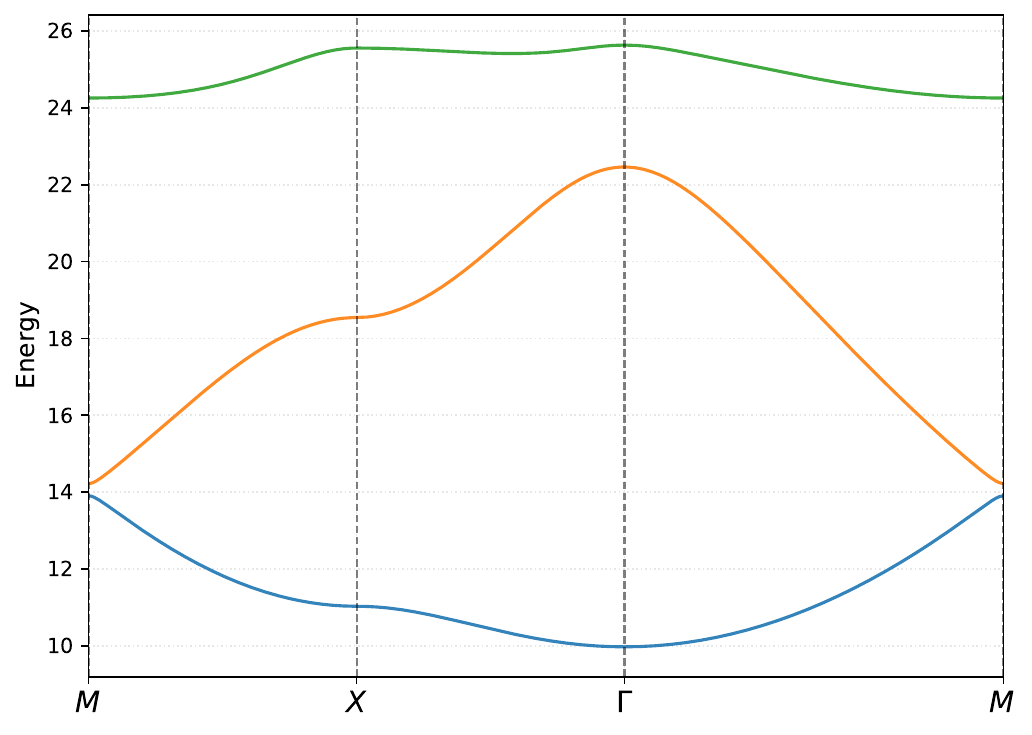}
        \caption{$V_0/|B|=3$}
        \label{fig: non_degenerate_bands_3_0000} 
    \end{subfigure}\hfill 
    \begin{subfigure}[t]{0.48\textwidth}
        \centering
        \includegraphics[width=\linewidth]{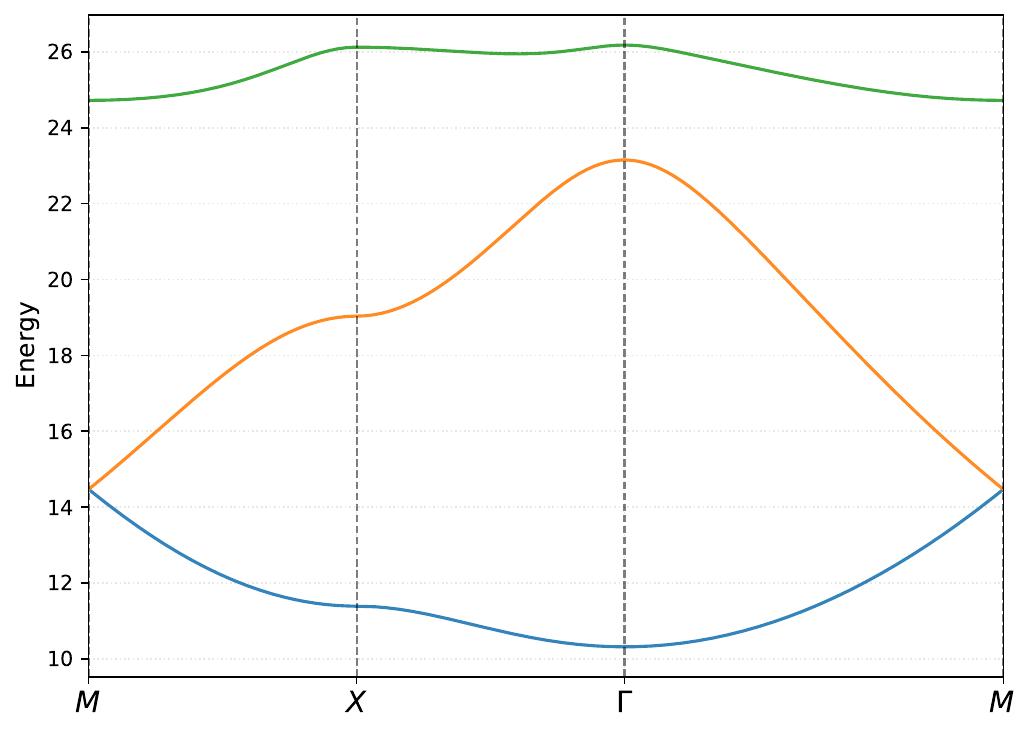}
        \caption{$V_0/|B|=3.1683$}
        \label{fig: non_degenerate_bands_3_1683} 
    \end{subfigure}
    
    \begin{subfigure}[t]{0.48\textwidth}
        \centering
        \includegraphics[width=\linewidth]{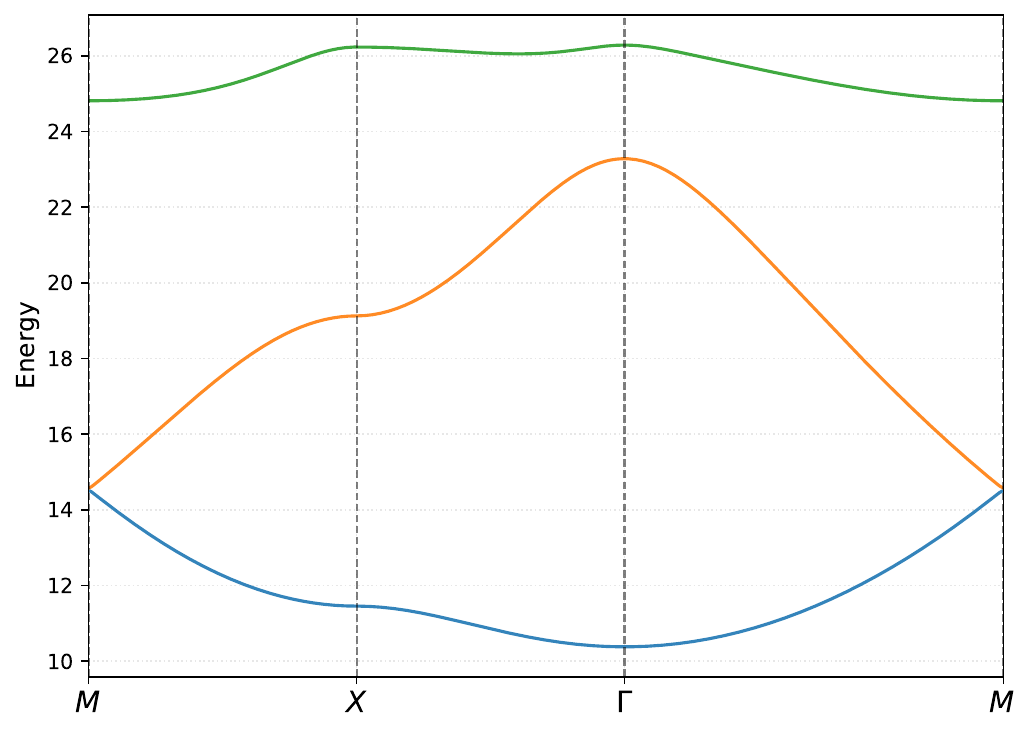}
        \caption{$V_0/|B|=3.2$}
        \label{fig: non_degenerate_bands_3_2000}
    \end{subfigure}\hfill 
    \begin{subfigure}[t]{0.48\textwidth}
        \centering
        \includegraphics[width=\linewidth]{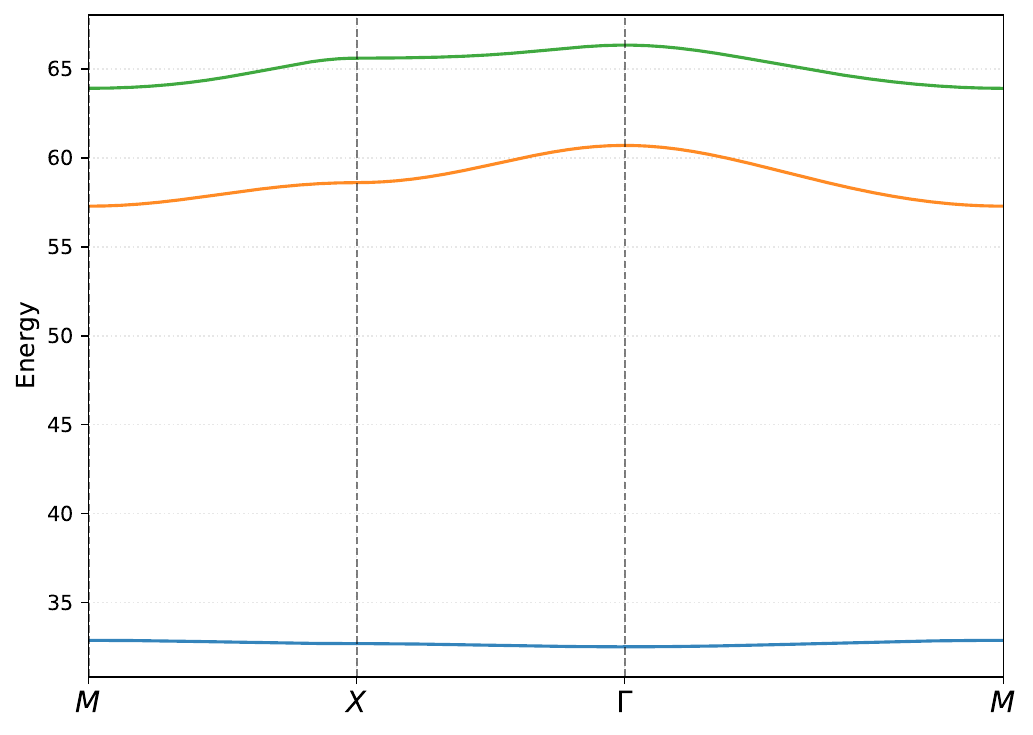}
        \caption{$V_0/|B|=20$}
        \label{fig: non_degenerate_bands_20_0000}
    \end{subfigure}

    \caption{Band structure for the non-degenerate system ($\beta=1$), generated by fixing the magnetic field strength to $B=2\pi$. Notice the band crossing between the lowest and the first excitation bands as the potential strength increases.}
    \label{fig: non_degenerate_bands}
\end{figure}

\begin{figure}[htbp]
    \centering
    \begin{subfigure}[t]{0.45\textwidth}
        \centering
        \includegraphics[width=\linewidth]{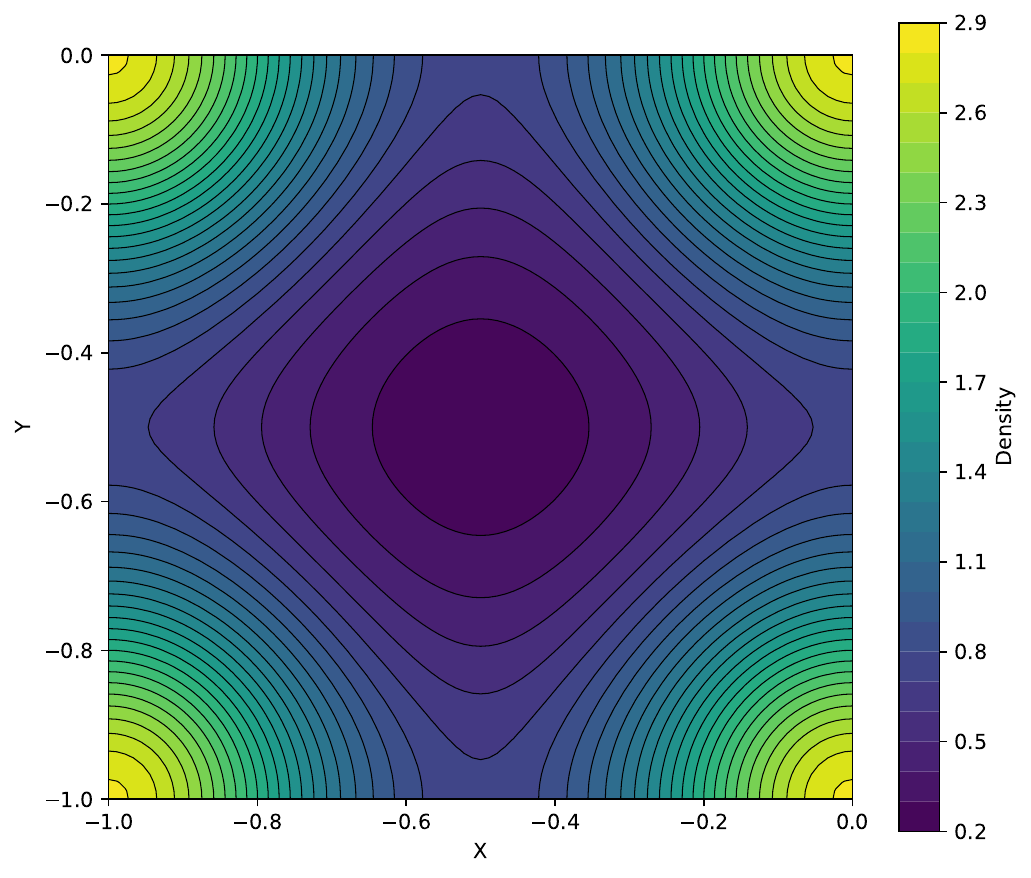}
        \caption{$V_0/|B| \approx 3.168$}
        \label{fig: non_degenerate_density_3_1683}
    \end{subfigure}\hfill 
    \begin{subfigure}[t]{0.45\textwidth}
        \centering
        \includegraphics[width=\linewidth]{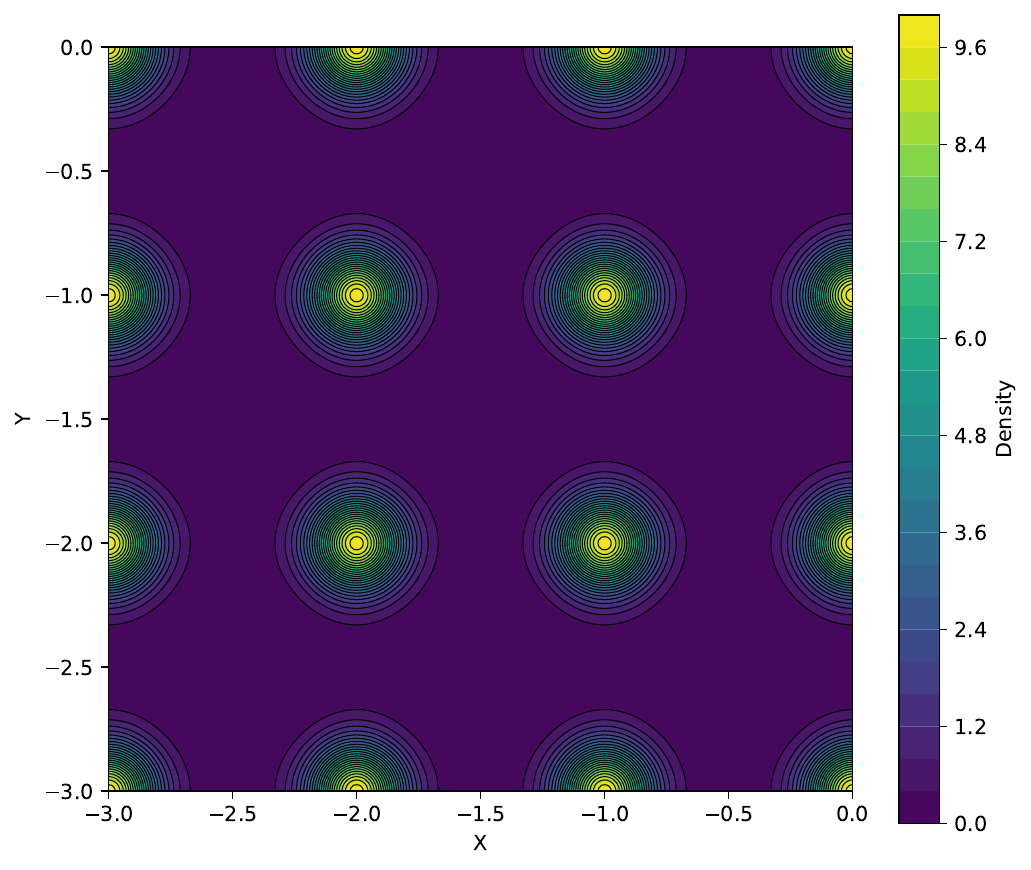}
        \caption{$V_0/|B|=20$}
        \label{fig: non_degenerate_density_20_0000}
    \end{subfigure}
    
    \caption{Ground-state density of the non-degenerate system ($\beta=1$, $B=2\pi$) under different potential strengths. Panel (a) shows the density at a moderate potential near the topological band crossing. Panel (b) shows the density evaluated over a $3 \times 3$ supercell for a large potential ratio, illustrating firm electron localization at the minima of the potential.}
    \label{fig: non_degenerate_density}
\end{figure}

\begin{table}[htbp]
\centering
\caption{Bisection method convergence log for the non-degenerate system ($\beta = 1$) showing the final 5 iterations.}
\label{tab:bisection_non_degenerate}
    \begin{tabular}{ccccccc}
        \toprule
        & \multicolumn{3}{c}{$R = \frac{V_0}{|B|}$} & \multicolumn{3}{c}{Chern Number $\mcC$} \\
        \cmidrule(lr){2-4} \cmidrule(lr){5-7}
        $i$ & Lower & Midpoint & Upper & $\mcC_l$ & $\mcC_m$ & $\mcC_u$ \\
        \midrule
        14 & 3.168311 & 3.168323 & 3.168335 & 1 & 1 & 0 \\
        15 & 3.168323 & 3.168329 & 3.168335 & 1 & 1 & 0 \\
        16 & 3.168329 & 3.168332 & 3.168335 & 1 & 1 & 0 \\
        17 & 3.168332 & 3.168333 & 3.168335 & 1 & 1 & 0 \\
        18 & 3.168333 & 3.168334 & 3.168335 & 1 & 1 & 0 \\
        \bottomrule
    \end{tabular}
\end{table}

\begin{figure}[htbp]
    \centering
    \begin{subfigure}[t]{0.30\textwidth}
        \centering
        \includegraphics[width=\linewidth]{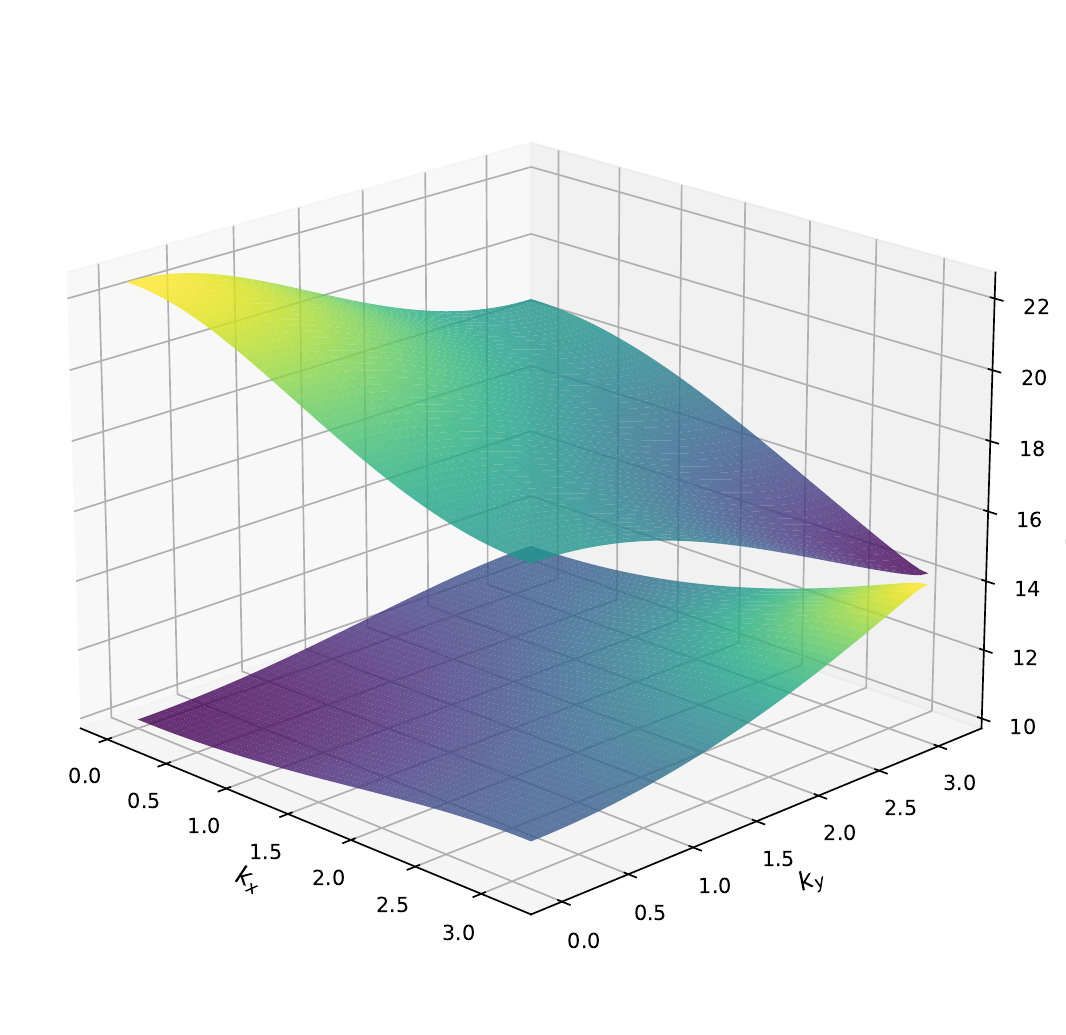}
        \caption{$V_0/|B|=3$}
        \label{fig: non_degenerate_surface_3_0000}
    \end{subfigure}\hfill 
    \begin{subfigure}[t]{0.30\textwidth}
        \centering
        \includegraphics[width=\linewidth]{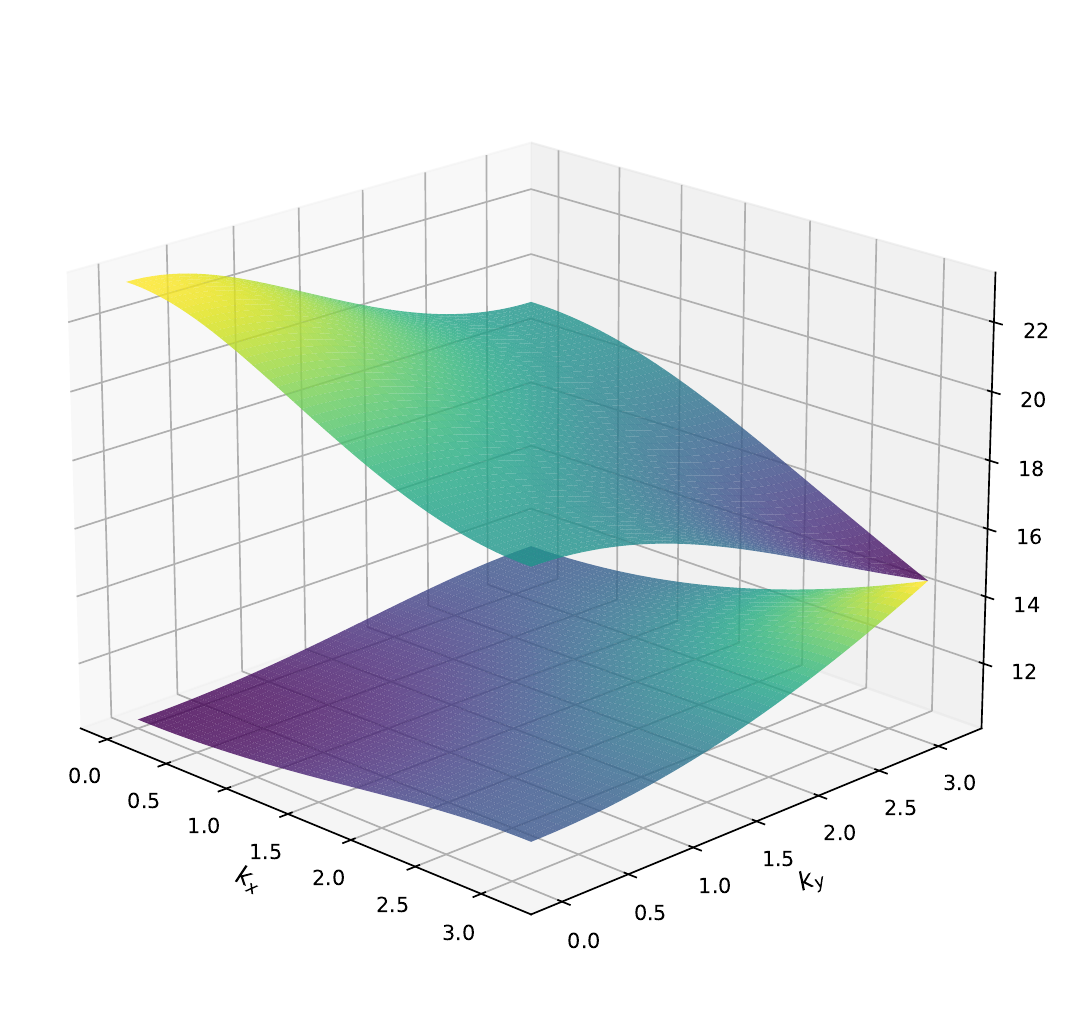}
        \caption{$V_0/|B| \approx 3.168$}
        \label{fig: non_degenerate_surface_3_1683}
    \end{subfigure}\hfill 
    \begin{subfigure}[t]{0.30\textwidth}
        \centering
        \includegraphics[width=\linewidth]{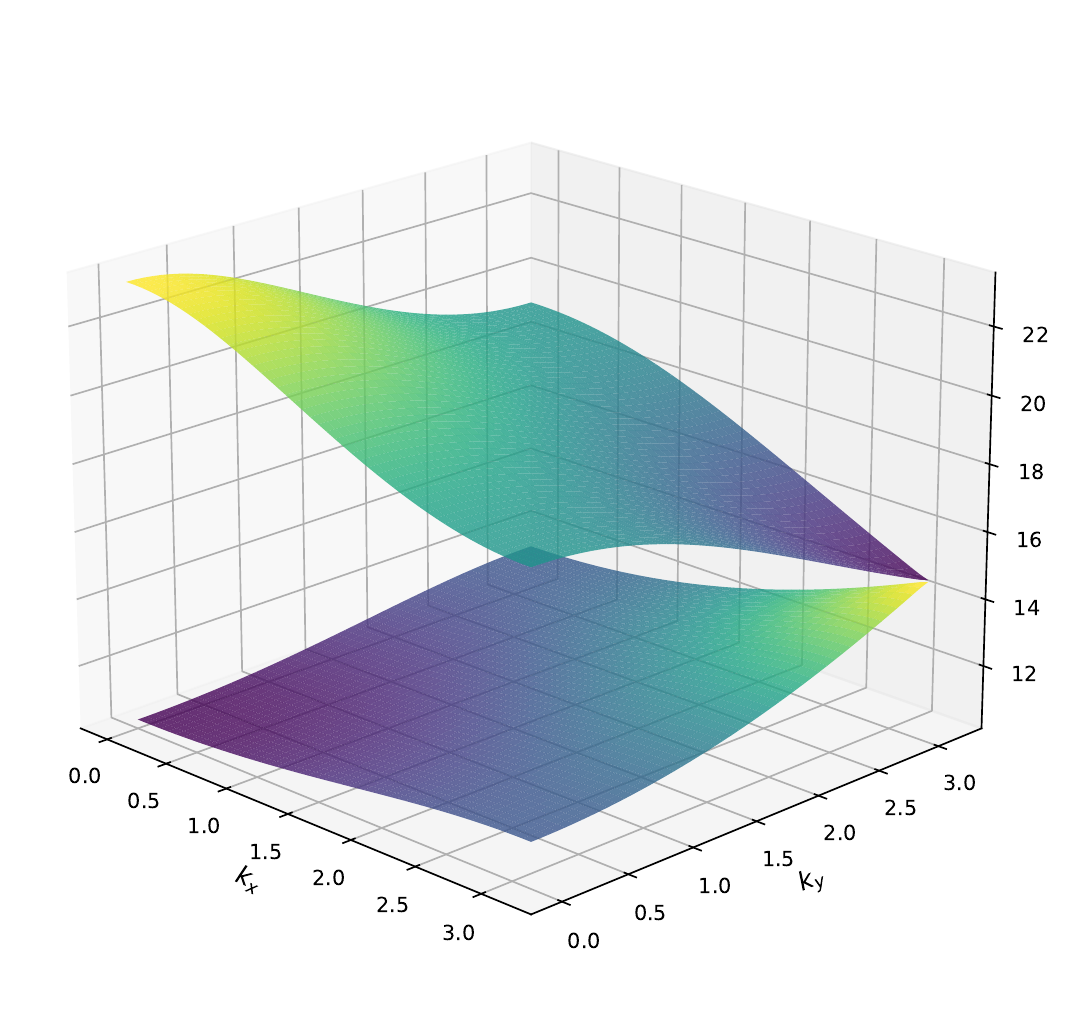}
        \caption{$V_0/|B|=3.2$}
        \label{fig: non_degenerate_surface_3_2000}
    \end{subfigure}
    
    \caption{Eigensurface of the non-degenerate system ($\beta=1$, $B=2\pi$) for three different potential strengths, illustrating the topology before and after the band crossing. Panel (a) precedes the crossing, panel (b) captures the system at the approximate crossing point, and panel (c) shows the stabilized eigensurface after the transition.}
    \label{fig: non_degenerate_surface}
\end{figure}

\begin{figure}[htbp]
\centering   
    \begin{subfigure}[t]{0.45\textwidth}
        \centering
        \includegraphics[width=\linewidth]{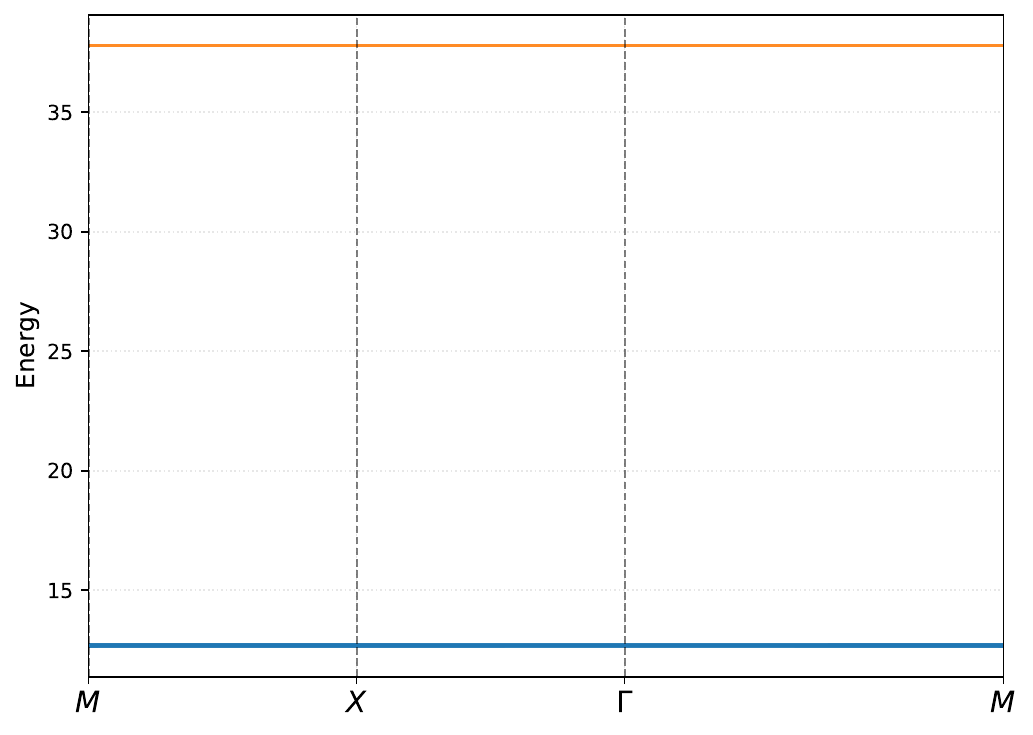}
        \caption{$V_0/|B|=0.01$}
        \label{fig: split_bands_0_0100} 
    \end{subfigure}\hfill 
    \begin{subfigure}[t]{0.45\textwidth}
        \centering
        \includegraphics[width=\linewidth]{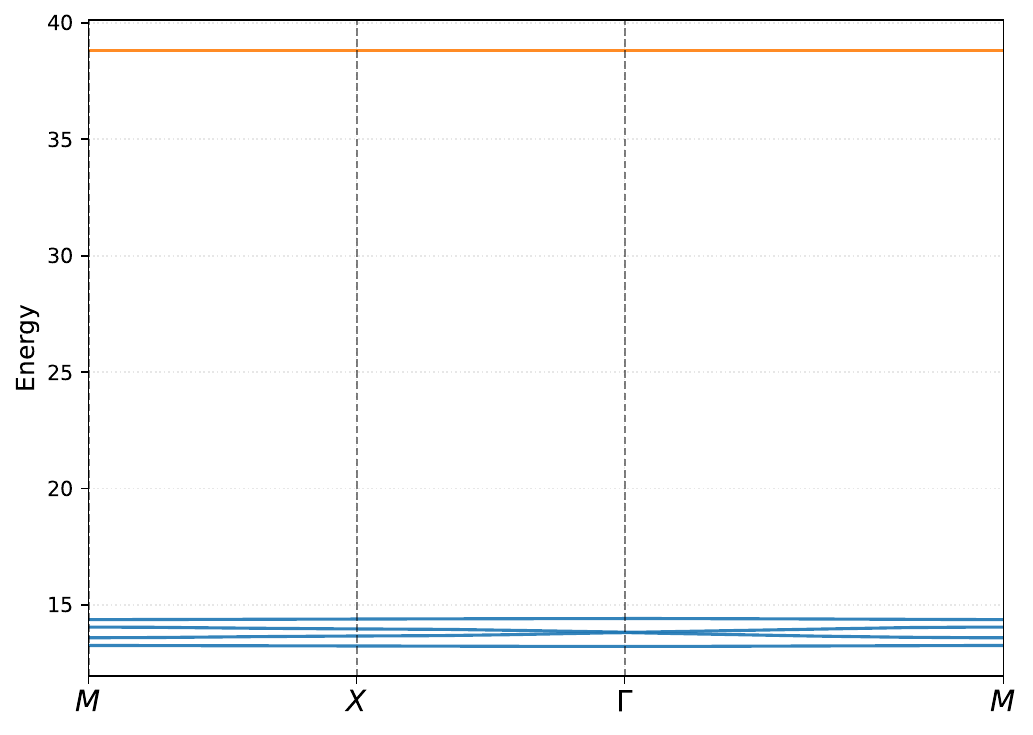}
        \caption{$V_0/|B|=0.1$}
        \label{fig: split_bands_0_1000} 
    \end{subfigure}

    \begin{subfigure}[t]{0.45\textwidth}
        \centering
        \includegraphics[width=\linewidth]{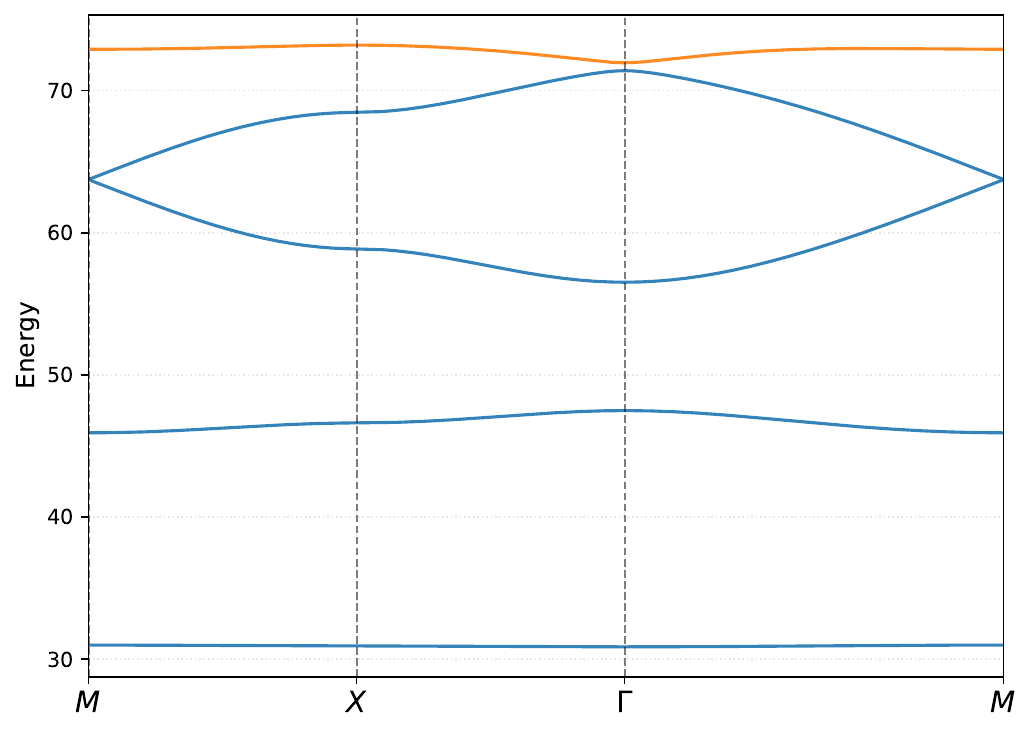}
        \caption{$V_0/|B|=3.8$}
        \label{fig: split_bands_3_8000} 
    \end{subfigure}\hfill 
    \begin{subfigure}[t]{0.45\textwidth}
        \centering
        \includegraphics[width=\linewidth]{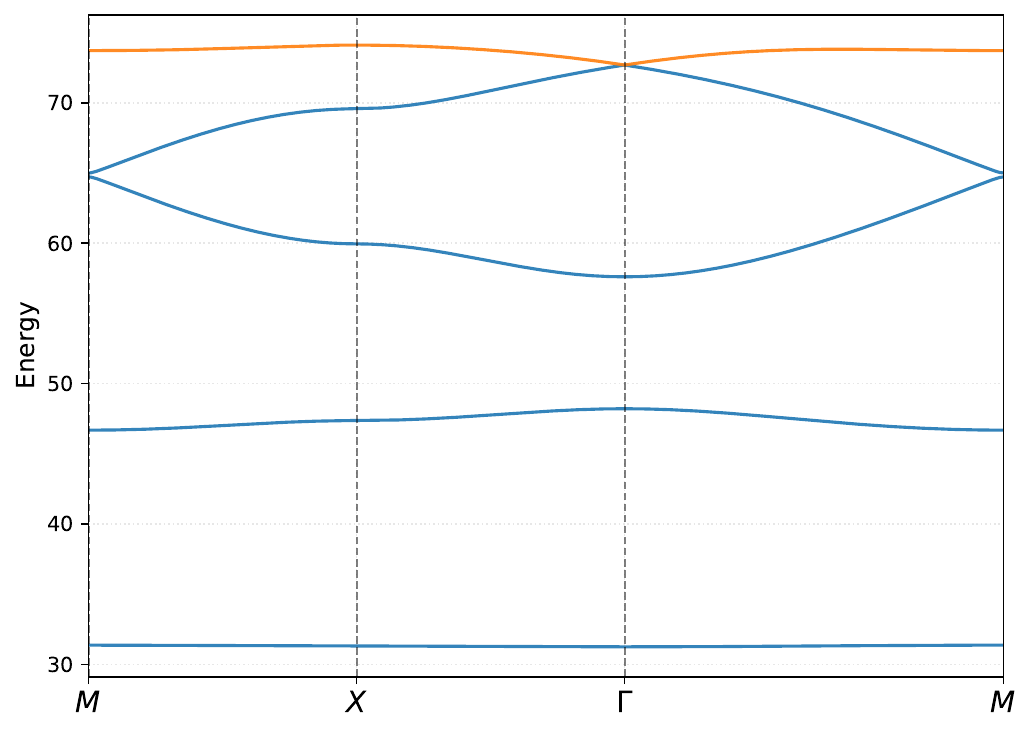}
        \caption{$V_0/|B| \approx 3.904$}
        \label{fig: split_bands_3_9036} 
    \end{subfigure}
    
    \begin{subfigure}[t]{0.45\textwidth}
        \centering
        \includegraphics[width=\linewidth]{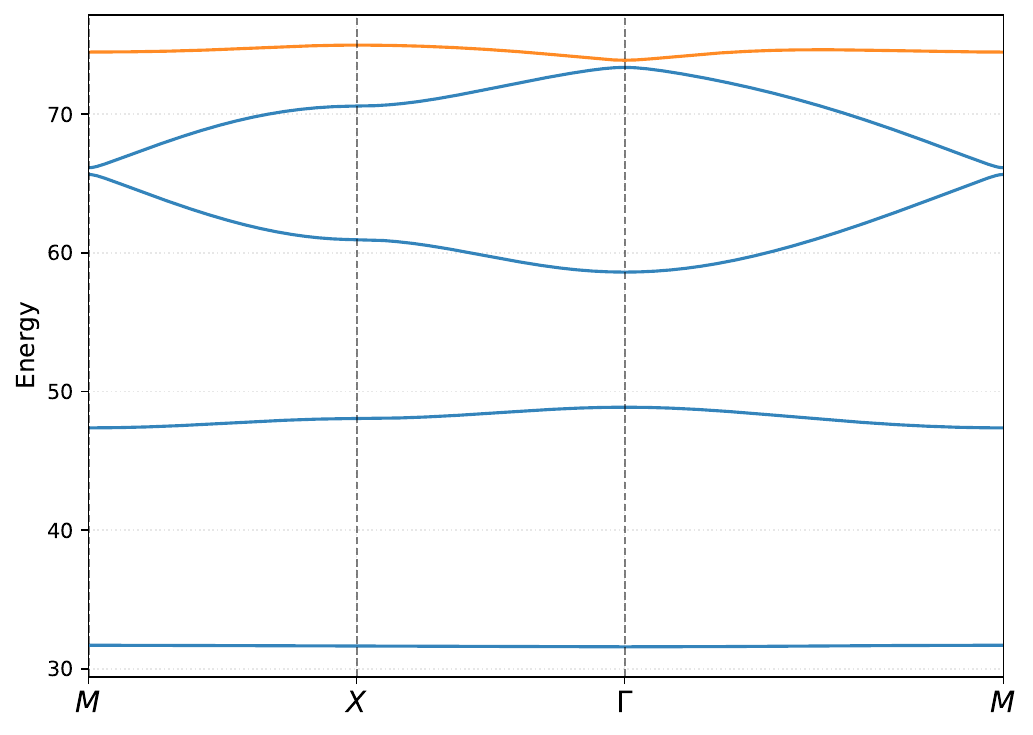}
        \caption{$V_0/|B|=4$}
        \label{fig: split_bands_4_0000}
    \end{subfigure}\hfill 
    \begin{subfigure}[t]{0.45\textwidth}
        \centering
        \includegraphics[width=\linewidth]{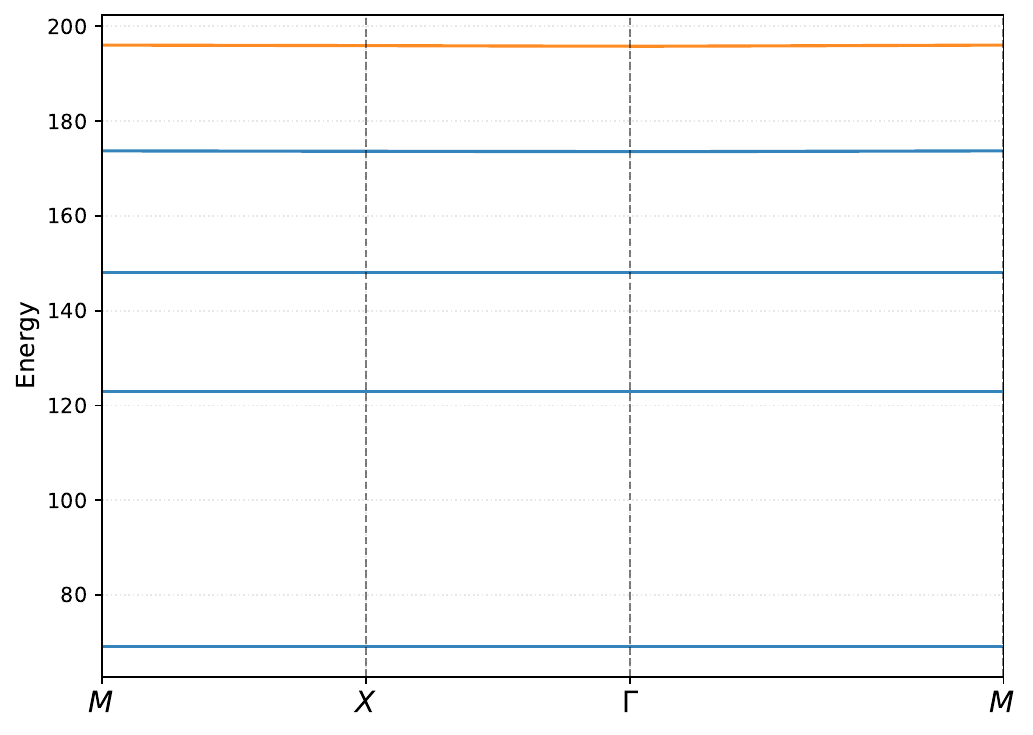}
        \caption{$V_0/|B|=20$}
        \label{fig: split_bands_20_0000}
    \end{subfigure}

    \caption{Band structure for the Split system with a degenerate basis ($\beta=4$). Notice the band splitting as the potential is introduced, and the band crossing between the third and fourth bands as the potential strength increases.}
    \label{fig: split_bands}
\end{figure}

\begin{table}[htbp]
\centering
\caption{Bisection method convergence log for the Split system ($\beta = 4$) showing the final 5 iterations.}
\label{tab:bisection_split}
    \begin{tabular}{ccccccc}
        \toprule
        & \multicolumn{3}{c}{$R = \frac{V_0}{|B|}$} & \multicolumn{3}{c}{Chern Number $\mcC$} \\
        \cmidrule(lr){2-4} \cmidrule(lr){5-7}
        $i$ & Lower & Midpoint & Upper & $\mcC_l$ & $\mcC_m$ & $\mcC_u$ \\
        \midrule
        14 & 3.903540 & 3.903552 & 3.903564 & 1 & 1 & 0 \\
        15 & 3.903552 & 3.903558 & 3.903564 & 1 & 0 & 0 \\
        16 & 3.903552 & 3.903555 & 3.903558 & 1 & 0 & 0 \\
        17 & 3.903552 & 3.903554 & 3.903555 & 1 & 0 & 0 \\
        18 & 3.903552 & 3.903553 & 3.903554 & 1 & 0 & 0 \\
        \bottomrule
    \end{tabular}
\end{table}

\begin{figure}[htbp]
    \centering
    \begin{subfigure}[t]{0.45\textwidth}
        \centering
        \includegraphics[width=\linewidth]{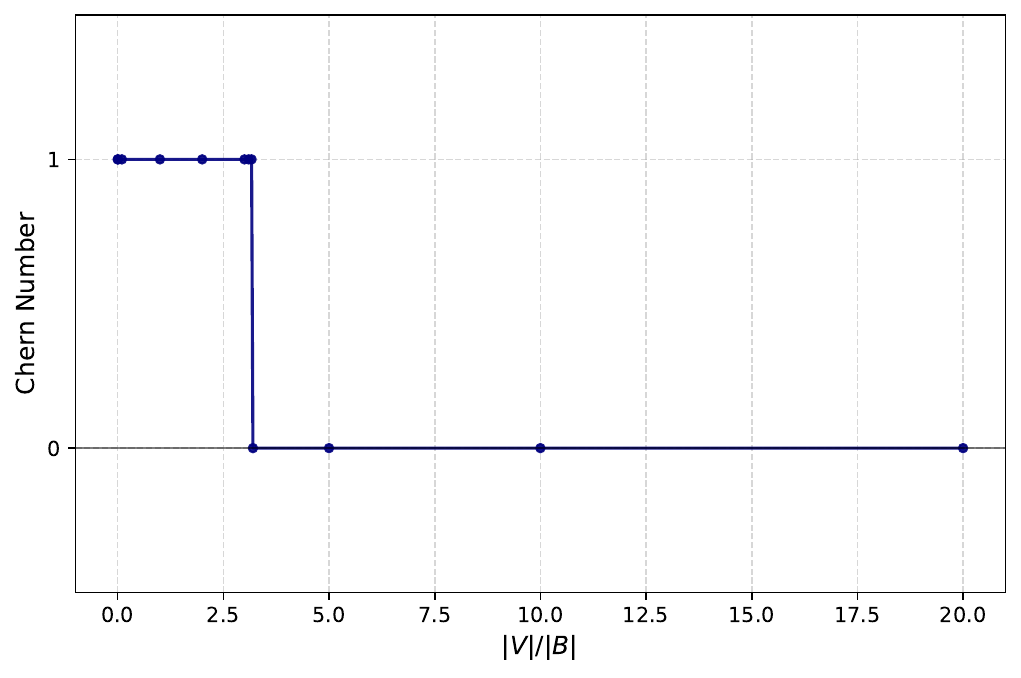}
        \caption{non-degenerate system ($\beta=1$)}
        \label{fig: non_degenerate_chern}
    \end{subfigure}\hfill 
    \begin{subfigure}[t]{0.45\textwidth}
        \centering
        \includegraphics[width=\linewidth]{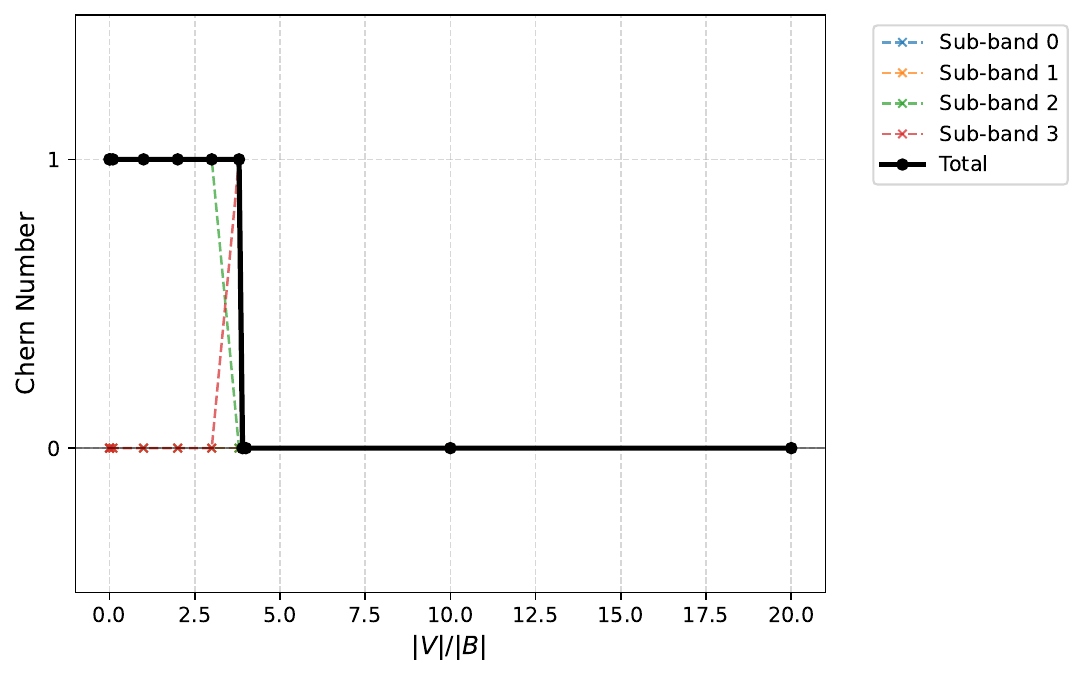}
        \caption{Split system ($\beta=4$)}
        \label{fig: split_chern}
    \end{subfigure}
    \caption{Chern number phase diagrams as a function of external potential to magnetic field strength ratio ($R$). The magnetic fields were kept fixed at $|B|=2\pi$ (left) and $|B|=8\pi$ (right).}
    \label{fig: combined_chern}
\end{figure}

\FloatBarrier

\subsection{Rational Fluxes}
We now consider the problem with a sinusoidal potential:
\begin{equation}
    V(\bfx) = \frac{V_0}{4}\left[2 - \cos(6\pi x_1) - \cos(6\pi x_2)\right].
\end{equation}
Due to the periodicity of $V$, the corresponding unit cell is $\Omega = \left[-1/3, 0\right] \times \left[-1/3, 0\right]$. For this simulation, we employed the same numerical discretization parameters detailed previously in Table~\ref{tab: simulation parameters}, fixing the magnetic field strength to $B=6\pi$. This configuration yields a fractional magnetic flux of $\beta = \frac{1}{3}$. In order to apply the methodology developed above, we employ a supercell formulation, considering $\tilde{\Omega} = [-1,0] \times [-1,0]$, where the tilde ($\tilde{\cdot}$) indicates quantities pertaining to the supercell. Consequently, the flux through the supercell becomes $\tilde{\beta} = 3$. 

Due to the periodicity of the potential over the supercell, each state inherits a degeneracy of degree $3$. In Figure~\ref{fig: supercell_density}, the resulting electron densities for the system under different potential strengths are displayed. Notice that the periodicity of the density accurately reflects the periodicity of the applied potential. Furthermore, as the potential strength increases, the densities become sharply localized at the potential wells. 

The evolution of the band structure is depicted in Figure~\ref{fig: supercell_bands}. In this system, the total Chern number is computed using the lowest band, as well as the first and second excitation bands, taking the degeneracy into account. We label the individual Chern numbers of these bands as $\mcC_1$, $\mcC_2$, and $\mcC_3$, respectively, such that the total Chern number of the system is $\mcC = \mcC_1 + \mcC_2 + \mcC_3$. All three individual Chern numbers, along with the total $\mcC$, are plotted as a function of $V_0/|B|$ in Figure~\ref{fig: supercell_chern}. 

We observe three distinct topological regimes, reflected both in the band structure (Figure~\ref{fig: supercell_bands}) and in the behavior of the individual Chern numbers. In the weak potential regime, each band carries a Chern number of $1$, resulting in a total Chern number of $\mcC = 3$. As the potential increases, a transition into an intermediate regime occurs, indicated by a sudden topological phase change where $\mcC_2$ drops from $1$ to $-2$. This is observed in Figures~\ref{fig: supercell_bands_3_0000} and~\ref{fig: supercell_bands_4_0000}, where the highest subband of the ground state suddenly touches the lowest subband of the first excited state. Following this transition, the total Chern number $\mcC$ stabilizes at $0$; however, examination of the individual subbands reveals a second crossing, as seen in Figures~\ref{fig: supercell_bands_9_0000} and~\ref{fig: supercell_bands_10_0000}. At this crossing, the Chern numbers of these bands exchange values: $\mcC_2$ jumps from $-2$ back to $1$, while $\mcC_3$ decreases from $1$ to $-2$. At this point, the system begins collapsing into the tight-binding regime, and the individual Chern numbers remain constant.

\begin{figure}[htbp]
    \centering
    \begin{subfigure}[t]{0.3\textwidth}
        \centering
        \includegraphics[width=\textwidth]{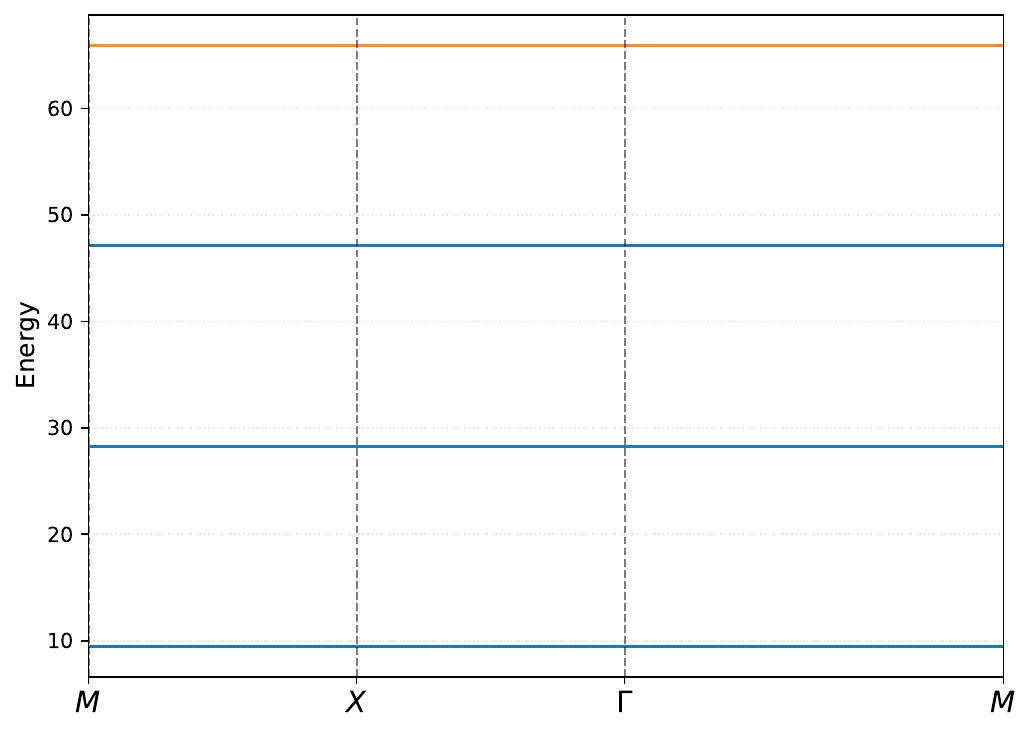}
        \caption{$V_0/|B|=0.01$}
        \label{fig: supercell_bands_0_0000}
    \end{subfigure}\hfill
    \begin{subfigure}[t]{0.3\textwidth}
        \centering
        \includegraphics[width=\textwidth]{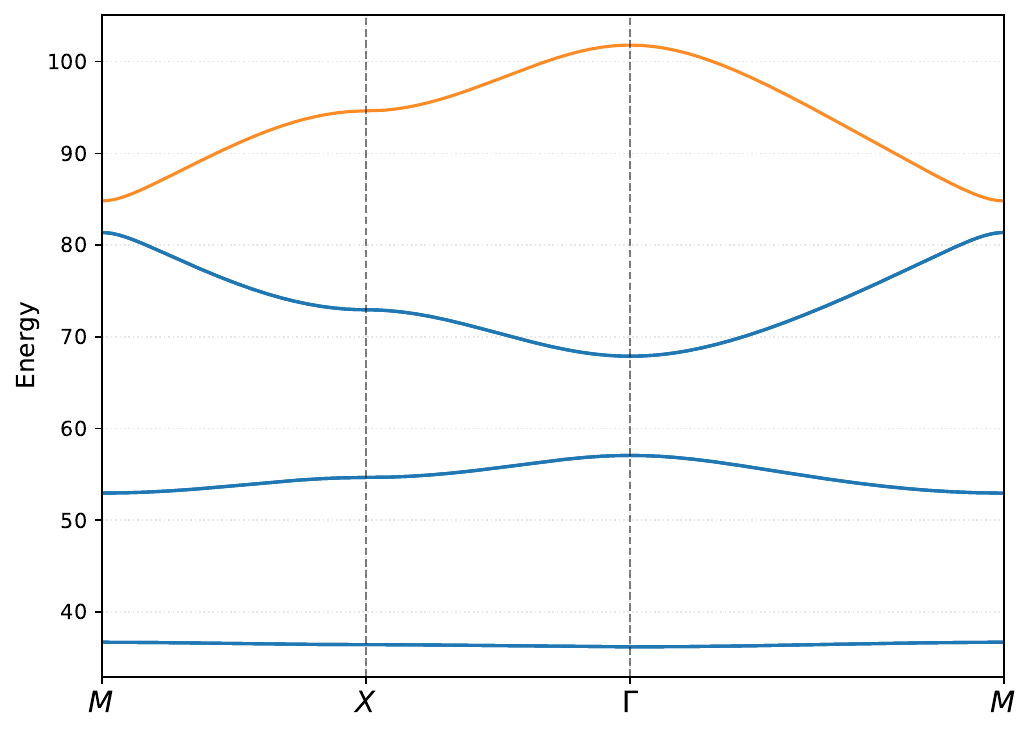}
        \caption{$V_0/|B|=3$}
        \label{fig: supercell_bands_3_0000}
    \end{subfigure}\hfill
    \begin{subfigure}[t]{0.3\textwidth}
        \centering
        \includegraphics[width=\textwidth]{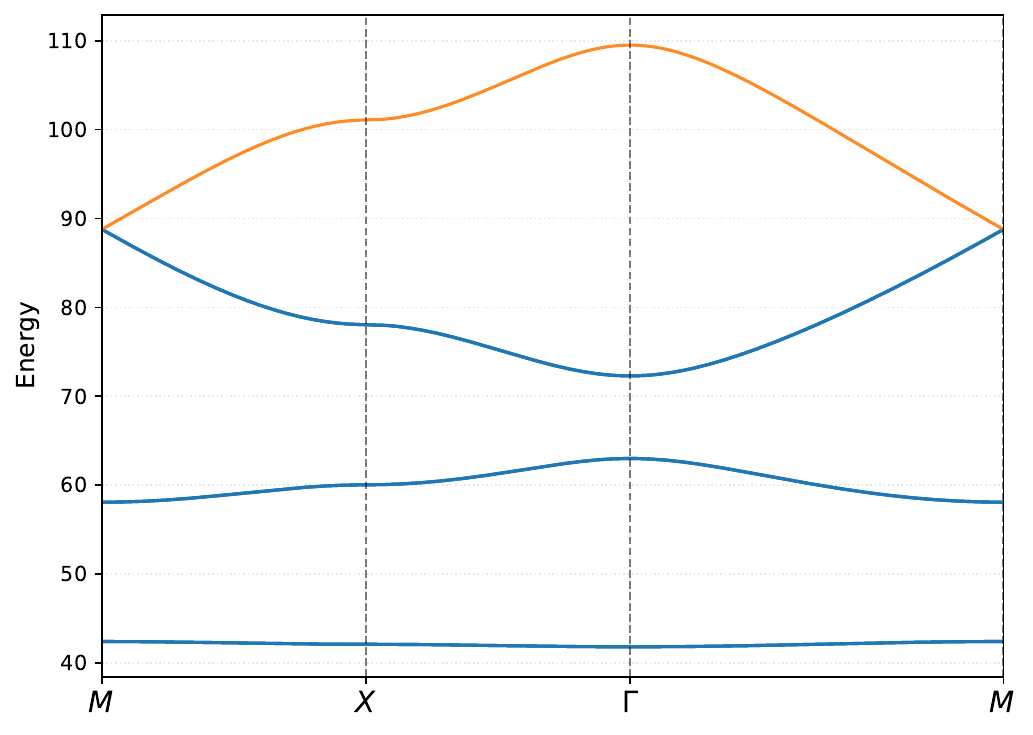}
        \caption{$V_0/|B| \approx 3.668$}
        \label{fig: supercell_bands_3_6676}
    \end{subfigure}
    
    \vspace{1em}
    
    \begin{subfigure}[t]{0.3\textwidth}
        \centering
        \includegraphics[width=\textwidth]{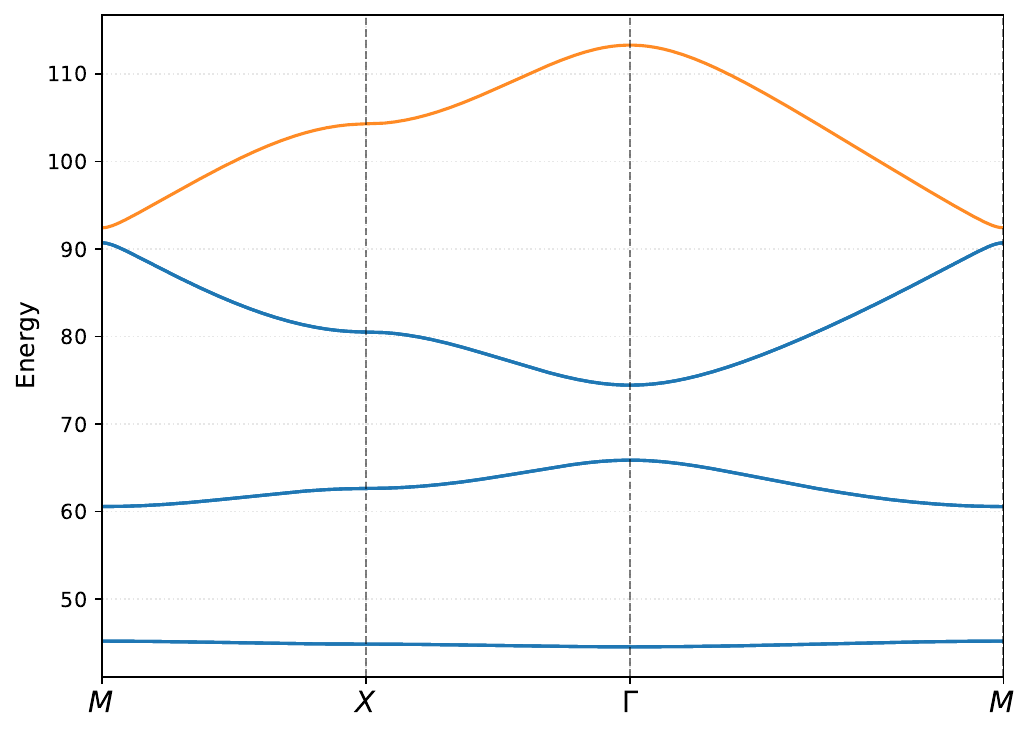}
        \caption{$V_0/|B|=4$}
        \label{fig: supercell_bands_4_0000}
    \end{subfigure}\hfill
    \begin{subfigure}[t]{0.3\textwidth}
        \centering
        \includegraphics[width=\textwidth]{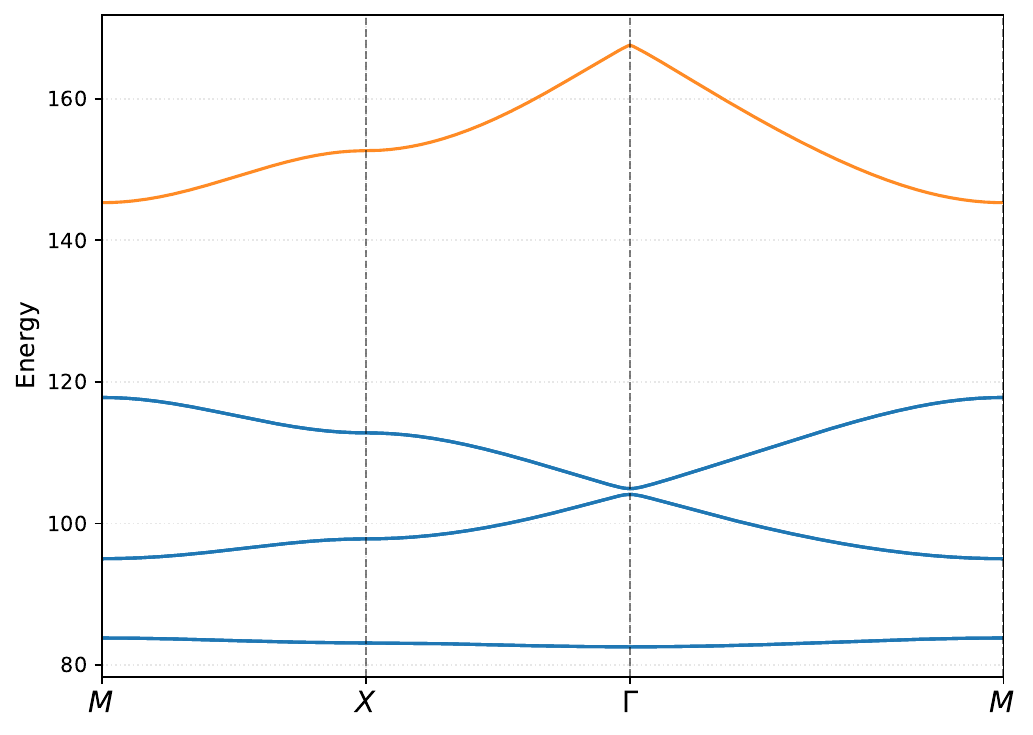}
        \caption{$V_0/|B|=9$}
        \label{fig: supercell_bands_9_0000}
    \end{subfigure}\hfill
    \begin{subfigure}[t]{0.3\textwidth}
        \centering
        \includegraphics[width=\textwidth]{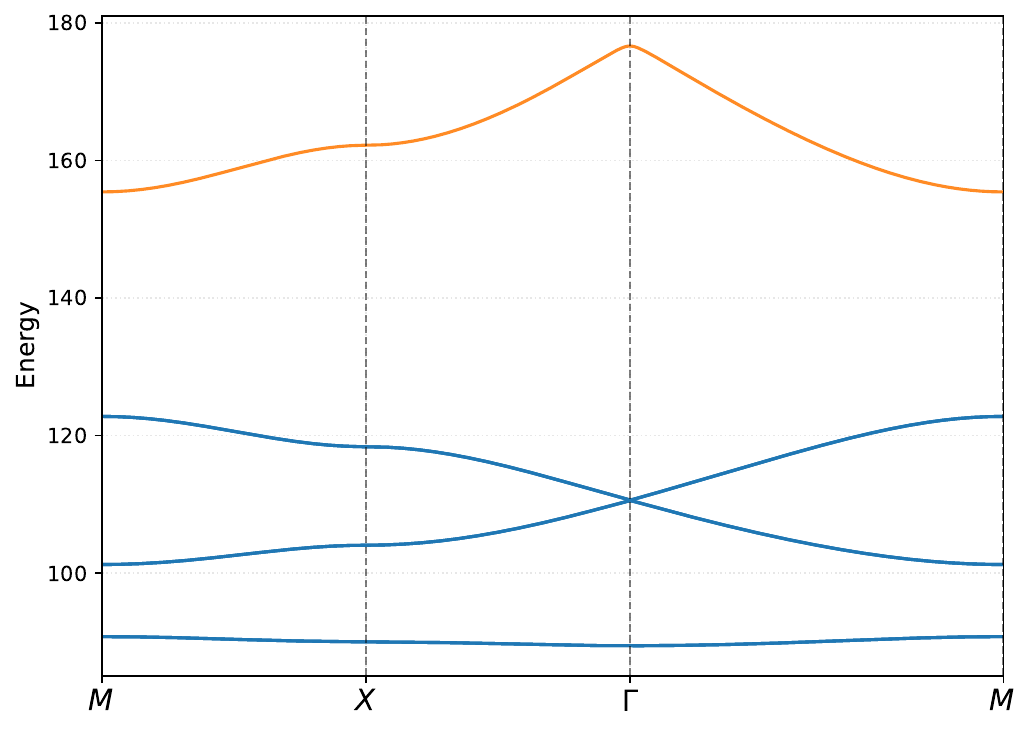}
        \caption{$V_0/|B|=10$}
        \label{fig: supercell_bands_10_0000}
    \end{subfigure}
    
    \caption{Band structure over the reciprocal supercell. Notice the collapse of the three lowest bands as the ratio between the potential strength and the magnetic field strength increases. The results were generated using a magnetic field strength of $B=6\pi$.}
    \label{fig: supercell_bands}
\end{figure}

\begin{figure}[htbp]
    \centering
    \includegraphics[width=0.55\textwidth]{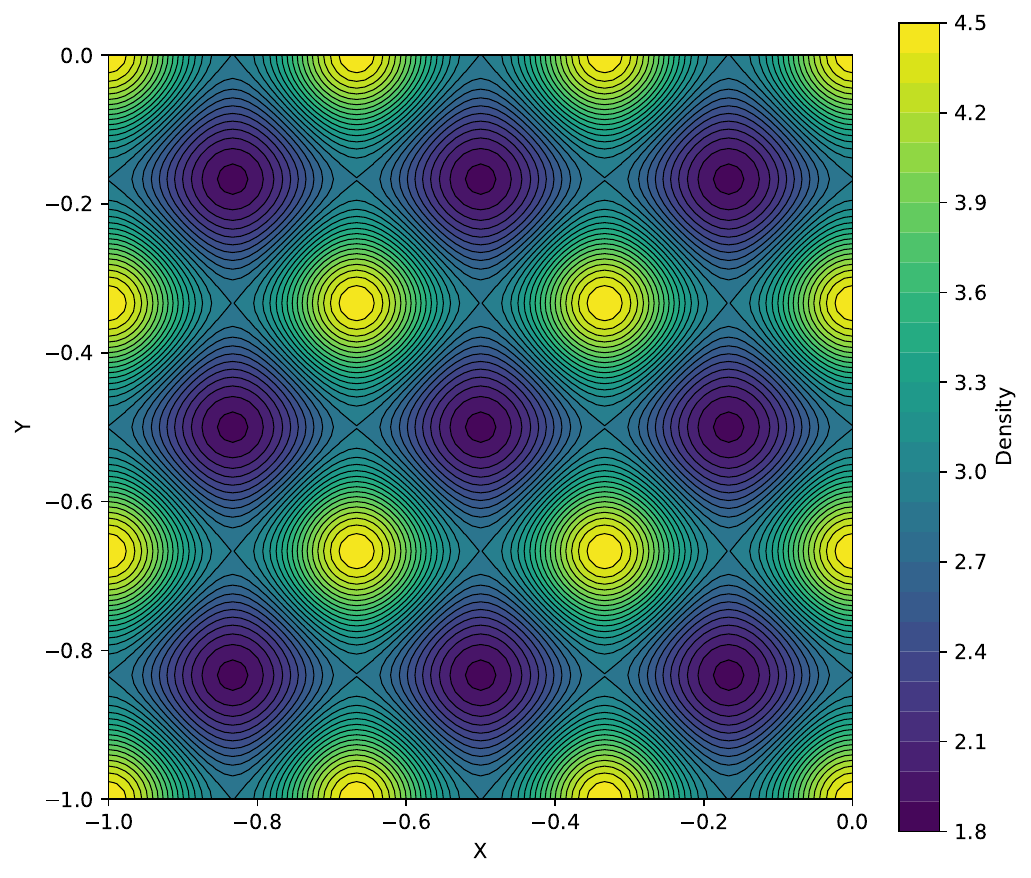}
    \caption{Ground-state density for the supercell formulation ($\beta=1/3$, yielding an effective supercell flux of $\tilde{\beta}=3$) near the first topological transition ($V_0/|B| \approx 3.668$). Notice the localization of the density at the minima of the potential.}
    \label{fig: supercell_density}
\end{figure}

\begin{figure}[htbp]
    \centering
    \includegraphics[width=0.50\textwidth]{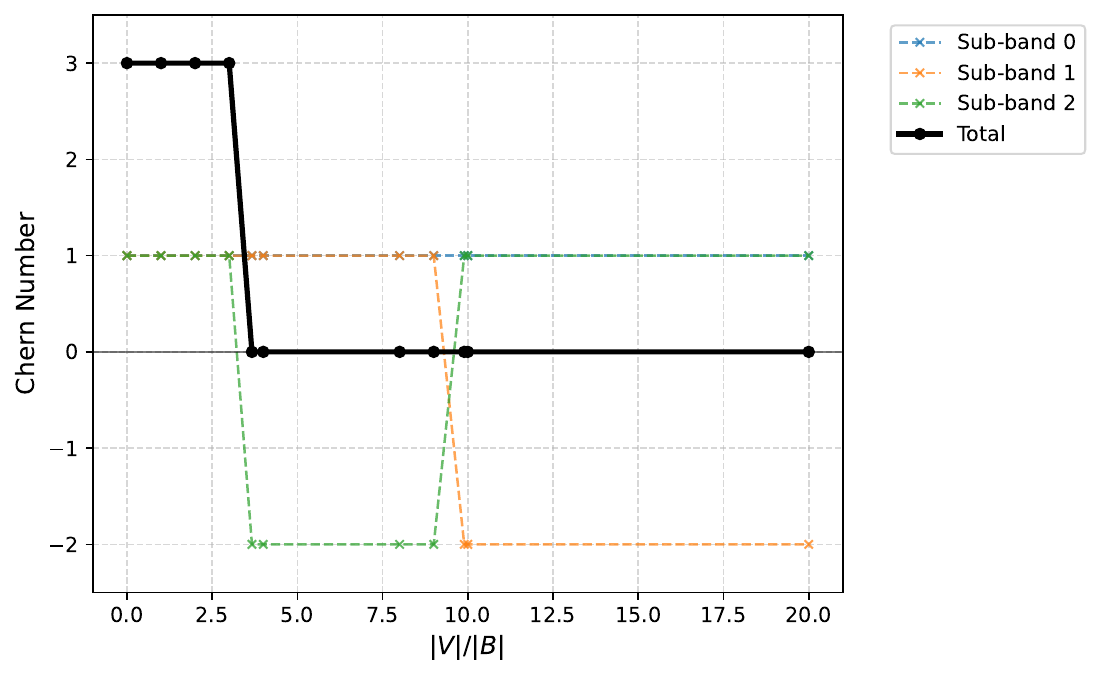}
    \caption{Phase diagram of the supercell system with a rational magnetic flux ($\beta=1/3$). The different regions are characterized by distinct topological phases indicated by their Chern numbers.}
    \label{fig: supercell_chern}
\end{figure}

%% file: Conclusions.tex
\section{Conclusions\label{sec: Conclusions}}

In this work, we have proposed a spectral method for the approximation of eigenfunctions corresponding to a Landau Hamiltonian subject to a periodic potential, using the eigenfunctions of the Landau operator as a basis. By carefully manipulating the expressions, we obtained a computationally feasible method. We also took advantage of the periodicity of the potential to expand it in a Fourier series, which allows us to employ the Fast Fourier Transform for the computations. We then proceeded to run benchmark simulations under a well-known optical lattice potential and explored the behavior of the bands as a function of the potential strength. 

The proposed method allows us to explore the behavior of a Landau system subject to a periodic potential in a computationally efficient way, and can be used to explore the behavior of the system under different potentials and to compute topological invariants such as the Chern number. In particular, we implemented the Fukui-Hatsugai-Suzuki (FHS) method to efficiently compute the Chern number of the lowest band for a wide range of potential strengths, recovering results consistent with the literature. The computation of the Chern number allowed us to assert the existence of band crossings. Furthermore, it is possible, by applying a simple bisection method, to narrow the range where said crossing occurs.

%% file: Data_availability_acknowledgements.tex
\section*{Data availability}
All the data and the code used to generate the results presented in this paper are available at \url{https://doi.org/10.18419/DARUS-6415}.

\section*{Acknowledgments}
RL and BS acknowledge support from the Deutsche Forschungsgemeinschaft (DFG, German Research Foundation) under project 516782692. We thank the Deutsche Forschungsgemeinschaft (DFG, German Research Foundation) for supporting this work by funding - EXC2075 – 390740016 under Germany’s
Excellence Strategy. We acknowledge the support by the Stuttgart Center for Simulation Science
(SimTech). The authors are grateful to the anonymous reviewers for their careful reading and valuable improvements of the paper.

The authors acknowledge the use of generative AI tools to improve the clarity and readability of both the manuscript the the corresponding code. The authors take full responsibility for the content of the publication.